\documentclass[11pt]{article}
\usepackage{amssymb,amsmath,amsthm,latexsym}
\usepackage{multicol}
\usepackage{float}
\usepackage{epsfig}
\usepackage{array}
\usepackage{delarray}

cm \textwidth 20.5 true cm \paperwidth 24.5 cm
\begin{document}
%\maketitle
\newtheorem{Theorem}{Theorem}
\newtheorem{Proposition}{Proposition}
\newtheorem{Lemma}{Lemma}
\newtheorem{Definition}{Definition}
\newtheorem{Corollary}{Corollary}
\newtheorem{Example}{Example}
\newtheorem{Remark}{Remark}
\newfont{\msbm}{msbm10 scaled\magstep1}
\newfont{\eufm}{eufm10 scaled\magstep1}
\newfont{\msam}{msam10 scaled\magstep1}
\numberwithin{equation}{section}
\def\R{\mbox{\msbm R}}
\def \e {\epsilon}
\def \ar {\rightarrow}
\def \exp {\rm exp}
\def \In   {\rm In}
\def \s {\sigma}

\title{On the de la Garza Phenomenon}

\author{ \sc Min Yang \\
University of Missouri}
\date{}
\maketitle \footnote {Research sponsored by NSF grants DMS-0707013
and DMS-0748409}

\begin{abstract}
Deriving optimal designs for nonlinear models is in general
challenging. One crucial step is to determine the number of
support points needed. Current tools handle this on a case-by-case
basis. Each combination of model, optimality criterion and
objective requires its own proof. The celebrated de la Garza
Phenomenon states that under a $(p-1)$th-degree polynomial
regression model, any optimal design can be based on at most $p$
design points, the minimum number of support points such that all
parameters are estimable. Does this conclusion also hold for
nonlinear models? If the answer is yes, it would be relatively
easy to derive any optimal design, analytically or numerically. In
this paper,  a novel approach is developed to address this
question. Using this new approach, it can be easily shown that the
de la Garza phenomenon exists for many commonly studied nonlinear
models, such as the Emax model, exponential model, three- and
four-parameter log-linear models, Emax-PK1 model, as well as many
classical polynomial regression models. The proposed approach
unifies and extends many well-known results in the optimal design
literature. It has four advantages over current tools: (i) it can
be applied to many forms of nonlinear models; to continuous or
discrete data; to data with homogeneous or non-homogeneous errors;
(ii) it can be applied to any design region; (iii) it can be
applied to multiple-stage optimal design; and (iv) it can be
easily implemented.

\end{abstract}


\section{Introduction}
The usefulness and popularity of nonlinear models have spurred a
large literature on data analysis, but research on design
selection has not kept pace.  One complication in studying optimal
designs for nonlinear models is that information matrices and
optimal designs depend on unknown parameters. A common approach to
solve this dilemma is to use locally optimal designs, which are
based on one's best guess of the unknown parameters. While a good
guess may not always be available, this approach remains of value
to obtain benchmarks for all designs (Ford, Torsney, and Wu,
1992). In fact, most available results are under the context of
locally optimal designs. (Hereafter, the word ``locally" is
omitted for simplicity.)

There is a vast literature on identifying good designs for a wide
variety of linear models, but the problem is much more difficult
and not nearly as well understood for nonlinear models. Relevant
references will be provided in later sections in this paper.

In the field of optimal designs, there exist no general approaches
for identifying good designs for nonlinear models. There are three
main reasons for this significant research gap. First, in
nonlinear models the mathematics tends to become more difficult,
which makes proving optimality of designs a more intricate
problem. Current available tools are mainly based on the geometric
approach by Elfving (1952) or the equivalence approach by Kiefer
and Wolfowitz (1960). This typically means that results can only
be obtained on a case-by-case basis. Each combination of model,
optimality criterion and objective requires its own proof. It is
not feasible to derive a general solution. Second, while linear
models are all of the form $E(y) = X\beta$, there is no simple
canonical form for nonlinear models. Coupled with the first
challenge, this means it is very difficult to establish unifying
and overarching results for nonlinear models. Again, this means
that individual consideration is typically needed for different
models, different optimality criteria, and different objectives.
Third, when considering the important practical problem of
multi-stage experiments, the search for optimal designs becomes
even more complicated because one needs to add design points on
top of an existing design.

Is there a practical way to overcome these challenges and derive a
general approach for finding optimal designs for nonlinear models?
One feasible strategy is to identify a subclass of designs with a
simple format, so that one can restrict considerations to this
subclass for any optimality problem. With a simple format, it
would be relatively easy to derive an optimal design, analytically
or numerically.

To make this strategy meaningful, the number of support points for
designs  in the subclass should be as small as possible.  By
Carath$\acute{\text{e}}$odory's theorem, we can always restrict
our consideration to at most $p(p+1)/2$  design points (where $p$
is the number of parameters). On the other hand, if we want all
parameters to be estimable, the minimum number of support points
should be at least $p$. Thus, the ideal situation is that the
designs in the subclass have no more than $p$ points. This reminds
one of de la Garza (1954)'s result, which was discussed in detail
by Pukelsheim (2006)  under the concept of ``admissibility". This
result was named  the celebrated de la Garza Phenomenon by Khuri,
Mukherjee, Sinha, and Ghosh (2006).

The de la Garza phenomenon can be explained as follows: Suppose we
consider a $(p-1)$th-degree polynomial regression model ($p$
parameters in total) with i.i.d. random errors. For any  $n$ point
design where $n>p$, there exists a design with exactly $p $
support points such that the information matrix of the latter one
is not inferior to that of the former one under Loewner ordering.
Does this phenomenon also exist for other models? For nonlinear
models with two parameters, Yang and Stufken (2009)  provided an
approach to identify the subclass of designs: for any design $\xi$
which does not belong to this class, there is a design in the
class with an information matrix that dominates $\xi$ in the
Loewner ordering. By applying this approach, they showed that many
commonly studied models, such as logistic and probit models, are
based on two design points. This result unifies and extends most
available optimality results for binary response models. However,
a limitation exists since it can only be applied to nonlinear
models with one or two parameters.

The purpose of this paper is to generalize Yang and Stufken (2009)
to nonlinear models with an arbitrary number of parameters.  The
proposed approach makes it relatively easy to prove the de la
Garza Phenomenon for many nonlinear models. In fact, for many
commonly studied nonlinear models, including the Emax model,
exponential model, three- and four-parameter log-linear models,
Emax-PK1 model, as well as many classical polynomial regression
models, it can be shown that for any given design $\xi$, there
exists a design $\xi^*$ with at most $p$ (number of parameters)
points, where  the information matrix under $\xi^*$ is not
inferior to that of $\xi$ under Loewner ordering. Thus, when
searching for an optimal design, one can restrict consideration to
this subclass of designs, both for one-stage and multi-stage
problems. Here, the optimal design can be for arbitrary parameter
functions under any information matrix based-optimality criterion,
including the commonly used $A$-, $D$-, $E$-, $\Phi_p$-, etc.
criteria as well as standardized optimality criteria proposed by
Dette (1997). Refer to Yang and Stufken (2009) for more details on
the significance of these flexibilities.

This paper is organized as follows. In Section 2, we introduce the
strategy. Main results are presented in Section 3.  Applications
to many commonly studied nonlinear models are presented in Section
4. Section 5 is a short discussion. Most proofs are included in
the Appendix.

\section{The strategy}
Suppose we have a nonlinear regression model for which at each point $x$ the experimenter observes a response $y$. We assume that
the $y$'s are independent and follow some exponential distribution $G$ with mean $\eta(x,\theta)$, where $\theta$ is $p\times 1$ parameters vector. Typically, the optimal nonlinear designs are studied under approximate theory, i.e.,  instead of exact sample sizes for design points, design weights are used. An approximate design $\xi$ can be written as  $\xi=\{(x_i,\omega_i), i=1,\ldots,n\}$, where $\omega_i>0$ is the design weight for design point $x_i$ and $\sum_{i=1}^n \omega_i=1$. It is more convenient to rewrite $\xi$ as $\xi=\{(c_i,\omega_i), i=1,\ldots,n\}$, where $c_i\in [A,B]$ may depend on $\theta$ and is one-to-one map of  $x_i\in[U,V]$. Typically, the information matrix for $\theta$ under design $\xi$ can be written as
\begin{equation}\begin{split}
I_{\xi}(\theta)=P(\theta) \left(\sum_{i=1}^n\omega_i  C(\theta, c_i)\right)(P(\theta))^T. \label{infor1}
\end{split}\end{equation}
where
\begin{equation}\begin{split}
C(\theta, c_i)= \begin{pmatrix}
 \Psi_{11}(c_i) &  \Psi_{12}(c_i) & \ldots & \Psi_{1p}(c_i)\\
 \Psi_{12}(c_i) &  \Psi_{22}(c_i) & \ldots  &\Psi_{2p}(c_i)\\
\vdots & \vdots & \ddots  & \vdots\\
 \Psi_{1p}(c_i) &  \Psi_{2p}(c_i) & \ldots  &\Psi_{pp}(c_i)
\end{pmatrix}
 \label{infor2}
\end{split}\end{equation}
Here,  $P(\theta)$ is a $p\times p$ nonsingular matrix that depends on the value of $\theta$ only. Notice that while $I_{\xi}(\theta)$ is fixed for given $\theta$ and $\xi$, there is flexibility on $P(\theta)$ and
$C(\theta, c_i)$. For many models, we can adjust $P(\theta)$ so that all $\Psi_{lt}$'s in (\ref{infor2}) are free of $\xi$ and $\theta$. Some examples of (\ref{infor1}) and (\ref{infor2}) are given in Section 4.

Under locally optimality context, for two given designs  $\xi=\{(c_i,\omega_i), i=1,\ldots,n\}$ and  $\xi^*=\{(\tilde{c}_j,\tilde{\omega}_j), j=1,\ldots,\tilde{n}\}$,  $I_{\xi}(\theta)\leq I_{\xi^*}(\theta)$ is equivalent to $\sum_{i=1}^n\omega_i  C(\theta, c_i)\leq \sum_{j=1}^{\tilde{n}}\tilde{\omega}_j  C(\theta, \tilde{c}_j)$ (here and elsewhere in this paper,
matrix inequalities are under the Loewner ordering). One strategy  to show  $I_{\xi}(\theta)\leq I_{\xi^*}(\theta)$ is to
prove that  the following equations hold:
\begin{equation}\begin{split}
\sum_{i=1}^n \omega_i\Psi_{lt}(c_i)= \sum_{j=1}^{\tilde{n}} \tilde{\omega}_j\Psi_{lt}(\tilde{c}_j),  \label{strategy:1}
\end{split}\end{equation}
for $1\leq l\leq t\leq p$ except for some $l=t$ (one or more)
\begin{equation}\begin{split}
\sum_{i=1}^n \omega_i\Psi_{ll}(c_i)\leq \sum_{j=1}^{\tilde{n}} \tilde{\omega}_j\Psi_{ll}(\tilde{c}_j). \label{strategy:2}
\end{split}\end{equation}
The development of the new tool is based on this strategy. Notice that Yang and Stufken (2009) used the same strategy for the $p=2$ case. However, the picture for a general $p$ is completely different. This is because when $p=2$,  the existence of $\xi^*$ can be based on the existence of one $\tilde{c}$ and one $\tilde{\omega}$ satisfying two nonlinear equations, which can be solved explicitly. For an arbitrary $p$, it is unlikely to derive such explicit expressions for $\xi^*$ since the existence of $\xi^*$ is based on the existence of multiple (approximately $p(p+1)/4$) $\tilde{c}$'s and $\tilde{\omega}$'s satisfying multiple (approximately $p(p+1)/2$) nonlinear equations.  Alternative approaches must be employed. In the next section, some new algebra results will be provided to address these needs.

\section{The approach}
In this section, we shall show that, under certain conditions, for a general nonlinear model, there exists a subclass of designs such that for any given design $\xi$, there exists a design $\tilde{\xi}$ in this subclass such that $I_{\xi}(\theta)\leq I_{\xi^*}(\theta)$. We first introduce some new algebra results.
\subsection{Algebra results}
Let $\Psi_1, \ldots,
\Psi_k$ be $k$ functions defined on $[A,B]$. Throughout this paper, we have the following assumptions:

\textbf{Assumption:}
\begin{itemize}
\item[(i)] $\Psi_1, \ldots, \Psi_k$ are infinity differentiable;
\item[(ii)] $f_{l,l}$ has no zero value on $[A,B]$.
\end{itemize}
Here, $f_{l,t}$, $1\leq t\leq k; t\leq l\leq k$ are defined as follows:
\begin{equation}\begin{split}
f_{l,t}(c)=\left\{
  \begin{array}{ll}
    \Psi_l'(c), & \hbox{$t=1$, $l=1,\ldots,k$} \\
    \left(\frac{f_{l,t-1}(c)}{f_{t-1,t-1}(c)}\right)', & \hbox{$2\leq t\leq k$, $t\leq l\leq k$.}
  \end{array}
\right. \label{def:df}
\end{split}\end{equation}

The structure of computations of
$f_{l,t}$ can be viewed as the following lower triangular matrix.

\begin{equation}\begin{split}
\begin{pmatrix}
  f_{1,1}=\Psi_1' &  &  &  &   \\
  f_{2,1}=\Psi_2' & f_{2,2}=\left(\frac{f_{2,1}}{f_{1,1}}\right)' &  &  &
  \\
  f_{3,1}=\Psi_3' & f_{3,2}=\left(\frac{f_{3,1}}{f_{1,1}}\right)'  & f_{3,3}=\left(\frac{f_{3,2}}{f_{2,2}}\right)' &  &
  \\
  f_{4,1}=\Psi_4' & f_{4,2}=\left(\frac{f_{4,1}}{f_{1,1}}\right)' & f_{4,3}=\left(\frac{f_{4,2}}{f_{2,2}}\right)'  & f_{4,4}=\left(\frac{f_{4,3}}{f_{3,3}}\right)'  &
\\
   \vdots & \vdots  & \vdots  & \vdots   & \ddots
\end{pmatrix} \label{df}
\end{split}\end{equation}

The $(t+1)$th column is obtained from the $t$th column. The $l$th ($l\geq t+1$) element of the $(t+1)$th column is the derivative of the ratio between the $l$'th  and the $t$'th element of the $t$th column.

\begin{Lemma} \label{exist}
Let $\Psi_1, \ldots,
\Psi_k$ be $k$ functions defined on $[A,B]$. Assume that $f_{l,l}(c)>0$, $c\in[A,B]$, $l=1,\ldots,k$. Then we have following conclusions:
\begin{itemize}
\item [(a)] when $k=2n-1$. For any given  $A\leq \tilde{c}_0<c_1<\ldots <c_n\leq B$ and $\omega_i>0$, $i=1,\ldots,n$, there exist n pairs $(\tilde{c}_j,\tilde{\omega}_j)$, $j=0,\ldots,n-1$, where $\tilde{c}_0<c_1<\tilde{c}_1<c_2 <\ldots <\tilde{c}_{n-1}<c_n$ and $\tilde{\omega}_j>0$, such that (\ref{exist:1}) and (\ref{exist:2}) hold, and
(\ref{exist:3})$>0$.

\item [(b)] when $k=2n-1$. For any given  $A\leq c_1<\ldots <c_n<\tilde{c}_n\leq B$ and $\omega_i>0$, $i=1,\ldots,n$, there exist n pairs $(\tilde{c}_j,\tilde{\omega}_j)$, $j=1,\ldots,n$, where $A\leq c_1<\tilde{c}_1<c_2 <\ldots <\tilde{c}_{n-1}<c_n<\tilde{c}_n\leq B$ and $\tilde{\omega}_j>0$, such that (\ref{exist:1}) and (\ref{exist:2}) hold, and
(\ref{exist:3})$<0$.

  \item [(c)] when $k=2n$. For any given $A\leq \tilde{c}_0<c_1<\ldots <c_n<\tilde{c}_n\leq B$ and $\omega_i>0$, $i=1,\ldots,n$, there exist $n+1$ pairs $(\tilde{c}_j,\tilde{\omega}_j)$, $j=0,\ldots,n$, where $A\leq \tilde{c}_0<c_1<\tilde{c}_1< \ldots <c_n<\tilde{c}_n \leq B$ and $\tilde{\omega}_j>0$, such that (\ref{exist:1}) and (\ref{exist:2}) hold, and
  (\ref{exist:3})$<0$.

  \item [(d)] when $k=2n$. For any given  $A\leq c_1<\ldots <c_{n+1} \leq B$ and $\omega_i>0$, $i=1,\ldots,n+1$, there exist $n$ pairs $(\tilde{c}_j,\tilde{\omega}_j)$, $j=1,\ldots,n$, where $A\leq c_1<\tilde{c}_1< \ldots <c_n<\tilde{c}_n <c_{n+1}\leq B$ and $\tilde{\omega}_j>0$, such that (\ref{exist:1}) and (\ref{exist:2}) hold, and
  (\ref{exist:3})$>0$.
\end{itemize}

Here,
\begin{eqnarray}
&&\sum_{i} \omega_i=\sum_{j}\tilde{\omega}_j; \label{exist:1}\\
&&\sum_{i}\omega_i\Psi_l(c_{i})=\sum_{j}\tilde{\omega}_j\Psi_l(\tilde{c}_{j}),
l=1,\ldots,k-1; \label{exist:2}\\
&&
\sum_{i}\omega_i\Psi_k(c_{i})-\sum_{j}\tilde{\omega}_j\Psi_k(\tilde{c}_{j}).\label{exist:3}
\end{eqnarray}
\end{Lemma}

Yang and Stufken (2009) has proven Lemma \ref{exist} for $k=2$ and
3. For arbitrary $k$, the proof is rather complicated (see
Appendix). Lemma \ref{exist} requires that $f_{l,l}(c)>0$ for
every $l=1,\ldots,k$, which is very demanding. In fact, such
strict conditions are not required. Suppose there are some
$f_{l,l}(c)<0$, we can consider  $-\Psi_l(c)$ instead of
$\Psi_l(c)$ depending on the situation, such that the
corresponding $f_{l,l}(c)>0$ for every $l=1,\ldots,k$. Notice that
(\ref{exist:2}) is invariant to such transformation and the sign
of the inequality (\ref{exist:3}) may need to be reversed. Thus we
can have similar results as Lemma \ref{exist} with a relaxed
condition. We are ready to present our first main theorem.

\begin{Theorem} \label{main1}
Let $\Psi_1, \ldots,
\Psi_k$ be $k$ functions defined on $[A,B]$. Let $F(c)=\prod_{l=1}^kf_{l,l}(c)$. For any given $N$ pairs $(c_i,\omega_i)$, where $c_i\in[A,B]$ and $\omega_i>0$, $i=1,\ldots,N$, there exists a set of pairs $(\tilde{c}_j,\tilde{\omega}_j)$, where $\tilde{c}_j\in[A,B]$ and $\tilde{\omega}_j>0$, such that  (\ref{exist:1}) and (\ref{exist:2}) holds, and
(\ref{exist:3})$<0$. Specifically,

\begin{itemize}
\item [(a)] when $k=2n-1$, $N\geq n$, and $F(c)<0$ for $c\in[A,B]$,  there are $n$ pairs $(\tilde{c}_j,\tilde{\omega}_j)$ in the set and one of $\tilde{c}_j$'s  is $A$;

\item [(b)] when $k=2n-1$, $N\geq n$, and $F(c)>0$ for $c\in[A,B]$, there are $n$ pairs $(\tilde{c}_j,\tilde{\omega}_j)$ in the set and one of $\tilde{c}_j$'s  is $B$;

  \item [(c)] when $k=2n$, $N\geq n$, and $F(c)>0$ for $c\in[A,B]$, there are $n+1$ pairs $(\tilde{c}_j,\tilde{\omega}_j)$ in the set and two of $\tilde{c}_j$'s  are $A$ and $B$;

  \item [(d)] when $k=2n$, $N\geq n+1$, and $F(c)<0$ for $c\in[A,B]$, there are $n$ pairs $(\tilde{c}_j,\tilde{\omega}_j)$ in the set.
  \end{itemize}

\end{Theorem}

\begin{proof}
The proofs for the above four cases are completely analogous. Here we will provide the proof of Theorem \ref{main1} (a). First, we prove the conclusion holds when $N=n$. From (\ref{df}), it is easy to verify that if we change only one $\Psi_l(c)$ to $-\Psi_l(c)$, say, $l=l_0$, and keep all other $\Psi_l(c)$'s the same, then all $f_{l,l}(x)$ will maintain their original signs with two exceptions: (i) $f_{l_0,l_0}(c)$ and $f_{l_0+1,l_0+1}(c)$ reverse the sign when $l_0<k$ or (ii) $f_{k,k}(c)$  reverse the sign when $l_0=k$. Among all $f_{l,l}$, $l=1,\ldots,k$, suppose $a$ of them are negative, say $f_{l_1,l_1},\ldots,f_{l_a,l_a}$. Here, $l_1<\ldots<l_a$ and  $a$ must be an odd number.

When $l_{2b-1}\leq l <l_{2b}$ ($1\leq b \leq (a-1)/2$) or $l\geq l_a$, $\widetilde{\Psi}_l(c)$ is defined as $-\Psi_l(c)$. Otherwise, $\widetilde{\Psi}_l(c)=\Psi_l(c)$.
We can verify that the corresponding $\widetilde{f}_{l,l}(c)>0, l=1,\ldots,k$ by repeatedly using the argument for  the change of signs of $f_{l,l}(c)$'s when we change only one $\Psi_l(c)$ to $-\Psi_l(c)$ each time. Now let $\tilde{c}_0=A$, and notice that  $\widetilde{\Psi}_k(c)=-\Psi_k(c)$, by Lemma \ref{exist} (a), the conclusion follows.

Assume that Lemma \ref{exist} (a) holds for $n\leq N\leq M$. Now we consider $N=M+1$. Following this assumption, for the $M$ pairs $(c_i,\omega_i)$, $1\leq i\leq M$, there exist $n$ pairs $(\overline{c}_j,\overline{\omega}_j)$, $j=0,\ldots,n-1$, where $\overline{c}_0=A$, such that
\begin{equation}\begin{split}
&\sum_{i=1}^M \omega_i=\sum_{j=0}^{n-1}\overline{\omega}_j, \\
&\sum_{i=1}^M \omega_i\Psi_l(c_{i})=\sum_{j=0}^{n-1} \overline{\omega}_j\Psi_l(\overline{c}_{j}),
l=1,\ldots,k-1,\\
&
\sum_{i=1}^M  \omega_i\Psi_k(c_{i})<\sum_{j=0}^{n-1} \overline{\omega}_j\Psi_k(\overline{c}_{j}).\label{main1:1}
\end{split}\end{equation}
Consider the $n-1$ pairs $(\overline{c}_j,\overline{\omega}_j)$, $j=1,\ldots,n-1$ and $(c_{M+1},\omega_{M+1})$. Apply (a) when $N=n$, there exist $n$ pairs $(\tilde{c}_j,\tilde{\omega}_j)$, $j=0,\ldots,n-1$ where $\tilde{c}_0=A$, such that
\begin{equation}\begin{split}
&\omega_{M+1}+\sum_{j=1}^{n-1}\overline{\omega}_j=\sum_{j=0}^{n-1}\tilde{\omega}_j, \\
&\omega_{M+1}\Psi_l(c_{M+1})+\sum_{j=1}^{n-1}\overline{\omega}_j\Psi_l(\overline{c}_{j})=\sum_{j=0}^{n-1}\tilde{\omega}_j\Psi_l(\tilde{c}_j),
l=1,\ldots,k-1,\\
&
\omega_{M+1}\Psi_k(c_{M+1})+\sum_{j=1}^{n-1}\overline{\omega}_j\Psi_k(\overline{c}_{j})<\sum_{j=0}^{n-1}\tilde{\omega}_j\Psi_k(\tilde{c}_j).\label{main1:2}
\end{split}\end{equation}
Combining (\ref{main1:1}) and (\ref{main1:2}), we establish Lemma \ref{exist} (a) when $N=M+1$. By mathematical induction, the conclusion follows.
\end{proof}

\subsection{The main tools}
We are now ready to present our main tools.
\begin{Theorem}\label{main2}
For a nonlinear regression model, suppose the information matrix can be written as (\ref{infor1}) and $c_i\in[A,B]$.
Rename all distinct $\Psi_{lt}$, $1\leq l\leq t\leq p$ to $\Psi_1, \ldots,
\Psi_k$ such that (i) $\Psi_k$ is one of $\Psi_{ll}$, $1\leq l\leq p$ and (ii) there is no $\Psi_{lt}=\Psi_k$ for $l<t$. Let $F(c)=\prod_{l=1}^kf_{l,l}(c)$, $c\in[A,B]$. For any given design $\xi$, there exists a design $\tilde{\xi}$, such that $I_{\xi}(\theta)\leq I_{\tilde{\xi}}(\theta)$. Here, $\tilde{\xi}$ depends on different situations.
\begin{itemize}
\item [(a)] when $k$ is odd and $F(c)<0$, $\tilde{\xi}$ is based on at most $(k+1)/2$ points including point $A$.
\item [(b)] when $k$ is odd and $F(c)>0$, $\tilde{\xi}$ is based on at most $(k+1)/2$ points including point  $B$.

\item [(c)] when $k$ is even and $F(c)>0$, $\tilde{\xi}$ is based on at most $k/2+1$ points including points $A$ and $B$.

  \item [(d)] when $k$ is even and $F(c)<0$, $\tilde{\xi}$ is based on at most $k/2$ points.
\end{itemize}
\end{Theorem}
\begin{proof}
The proof for the four cases are completely analogous. Here we provide the proof of Theorem \ref{main2} (a).

By (\ref{infor1}) and the fact that $P(\theta)$ depends on
$\theta$ only, it is sufficient to show that $C_{\xi}(\theta)\leq
C_{\tilde{\xi}}(\theta)$. Let $\xi=\{(c_i,\omega_i),
i=1,\ldots,N\}$. If $N< n$, then we can just take
$\tilde{\xi}=\xi$. If $N\geq n$, by (a) of Theorem \ref{main1},
there exist $n$ paris $(\tilde{c}_j, \tilde{\omega}_j)$,
$j=0,\ldots, n-1$, where $\tilde{c}_0=A$, such that
(\ref{exist:1}), (\ref{exist:2}), and (\ref{exist:3})$<0$ holds.
Let $\tilde{\xi}=\{(\tilde{c}_j, \tilde{\omega}_j), j=0,\ldots,
n-1\}$. Direct computation shows that the diagonal elements of
$C_{\tilde{\xi}}(\theta)-C_{\xi}(\theta)$ are either 0 or greater
than 0, and the off-diagonal elements are all 0. Thus the
conclusion follows.
\end{proof}

\begin{Remark}
For Cases (a), (b), and (d) of Theorem \ref{main2}, the conclusions stay the same if the interval $[A,B]$ is not finite. For Case (a), $[A,B]$ can be replaced by $[A,\infty)$. In this situation, for any given design $\xi$, we can choose $B=\mathrm{Max}_{1\leq i\leq N} c_i$ and the same conclusion follows. Similarly, $[A,B]$ can be replaced  by $(-\infty,B]$ in Case (b) or $(-\infty,\infty)$ in Case (d).
\end{Remark}

\begin{Remark}
There are many different ways to rename all distinct $\Psi_{lt}$, $1\leq l\leq t\leq p$ to $\Psi_1, \ldots,
\Psi_k$. Not all orders can satisfy the requirements in Theorem \ref{main2}. However, as long as there exists one order of $\Psi_1, \ldots,
\Psi_k$ such that these requirements can be satisfied,  the conclusion holds. Notice that $\Psi_k$ must be one of  $\Psi_{ll}$, $1\leq l\leq p$.
\end{Remark}

\section{Applications}
Theorem \ref{main2} can be applied to many commonly studied statistical models. In fact, as we demonstrate next, for many models, any optimal design can be based on the minimum number of support points, i.e., number of support points such as all parameters are estimable. As we discussed earlier, this makes it much easier to study an optimal design. Theorem \ref{main2} works on the information matrix directly. It  is very general. It can be applied to any models, continuous or discrete data with homogeneous or non-homogeneous   error,  as long as the information matrix can be written as (\ref{infor1}). Here, we demonstrate its applications for the model
\begin{equation}\begin{split}
Y_{ij}=\eta(x_i,\theta)+\epsilon_{ij},  \label{model1}
\end{split}\end{equation}
where $\epsilon_{ij}$'s are i.i.d. $N(0,\sigma^2)$ with known $\sigma^2$, $x_i\in [L,U]$ is the design variable and $\theta$ is a $p\times 1$ parameter vector. Most commonly
studied models can be written as  (\ref{model1}). For a given design $\xi=\{(x_i,\omega_i), i=1,\ldots,N\}$,  the corresponding information matrix for $\theta$ can be written as
\begin{equation}\begin{split}
I_{\xi}(\theta)=\sum_{i=1}^N \omega_i \frac{\partial \eta(x_i,\theta)}{\partial \theta}\left(\frac{\partial \eta(x_i,\eta)}{\partial \theta}\right)^T.  \label{infor3}
\end{split}\end{equation}
Next, we apply Theorem \ref{main2} for some popular choices of $\eta(x,\theta)$. Notice that Theorem \ref{main2} is not limited to this model format. It can be applied to many other models. Some examples will be shown in Sections \ref{nparm} and \ref{loglinear}.

\subsection{Models with three parameters }
Dette, Bretz, Pepelyshev, and Pinheiro (2008) studied $E_{max}$, Exponential, and Log-linear models.  These models can be written in the form of  (\ref{model1}) with
\begin{equation}\begin{split}
\eta(x,\theta)=\begin{cases} \theta_0+\frac{\theta_1 x}{x+\theta_2} & \text{$E_{max}$},\\
\theta_0+\theta_1\exp(x/\theta_2) & \text{Exponential},\\
\theta_0+\theta_1\log(x+\theta_2) & \text{Log-linear.}
\end{cases} \label{model3:1}
\end{split}\end{equation}
Here, $x_i\in [L,U]\subset (0,\infty)$, $\theta_1>0$, and $\theta_2>0$. They showed that local MED-optimal designs (MED is defined as the smallest dose producing a practically relevant response) are either a two points design with low end point $L$, or a three points design with two end points $L$ and $U$. In fact, as the following theorem shows, any optimal design can be based on three points including one or two end points.
\begin{Theorem}\label{example1}
Under model (\ref{model3:1}), for an arbitrary design $\xi$, there exists a design $\xi^*$ with three support points such that  $I_{\xi}(\theta)\leq I_{\xi^*}(\theta)$. Specifically, the three points include the two end points $L$ and $U$ for the $E_{max}$ model; the upper end point $U$ for the exponential model; and  the two end points $L$ and $U$ for the log-linear model.
\end{Theorem}
\begin{proof}
We first consider the $E_{max}$ model. By some routine algebra, it can be shown that the information matrix can be written as  in the form of  (\ref{infor1}) with
\begin{equation}\begin{split}
P(\theta)=\begin{pmatrix}
1 & 0 & 0\\
\frac{1}{\theta_2}& -\frac{1}{\theta_2} & 0\\
\frac{1}{\theta_2^2}& -\frac{1}{\theta_2^2} & \frac{1}{\theta_1\theta_2}
\end{pmatrix}^{-1} \text{ and }
C(\theta,c_i)=\begin{pmatrix}
1 &  c_i &   c_i^2\\
 c_i&  c_i^2 &  c_i^3\\
 c_i^2&  c_i^3& c_i^4
\end{pmatrix},
\end{split}\end{equation}
where $c_i=1/(x_i+\theta_2)$.
Let $\Psi_1(c)=c$, $\Psi_2(c)=c^2$, $\Psi_3(c)=c^3$, and $\Psi_4(c)=c^4$, we can verify that the corresponding $f_{1,1}=1$, $f_{2,2}=2$, $f_{3,3}=3$, and $f_{4,4}=4$. By Case (c) of Theorem \ref{main2}, the conclusion follows.

The proofs for exponential and loglinear models are similar with different $P(\theta)$ and $C_{\xi}(\theta)$. For the exponential model,
\begin{equation}\begin{split}
P(\theta)=\begin{pmatrix}
1 & 0 & 0\\
0& 1 & 0\\
0& 0& -\frac{\theta_2}{\theta_1}
\end{pmatrix}^{-1} \text{ and }
C(\theta,c_i)=\begin{pmatrix}
1 & e^{c_i} & c_ie^{c_i}\\
e^{c_i} & e^{2c_i} & c_ie^{2c_i}\\
c_ie^{c_i} & c_ie^{2c_i} &c_i^2e^{2c_i}
\end{pmatrix},
\end{split}\end{equation}
where $c_i=x_i/\theta_2$. Let $\Psi_1(c)=e^c$, $\Psi_2(c)=ce^c$, $\Psi_3(c)=e^{2c}$, $\Psi_4(c)=ce^{2c}$, and $\Psi_5(c)=c^2e^{2c}$, we can verify that the corresponding $f_{1,1}=e^c$, $f_{2,2}=1$, $f_{3,3}=2e^c$, $f_{4,4}=1$, and
$f_{5,5}=2$. By Case (b) of Theorem \ref{main2}, the conclusion follows.

For the log-linear model,
\begin{equation}\begin{split}
P(\theta)=\begin{pmatrix}
1 & 0 & 0\\
0& -1 & 0\\
0& 0& \frac{1}{\theta_1}
\end{pmatrix}^{-1} \text{ and }
C(\theta,c_i)=\begin{pmatrix}
1 & \log(c_i) & c_i\\
\log(c_i)& \log^2(c_i) & c_i\log(c_i)\\
c_i & c_i\log(c_i) &c_i^2
\end{pmatrix},
\end{split}\end{equation}
where $c_i=1/(x_i+\theta_2)$. Let $\Psi_1(c)=\log(c)$, $\Psi_2(c)=c$, $\Psi_3(c)=c\log(c)$, and $\Psi_4(c)=\log^2(c)$ or $c^2$, we can verify that the corresponding $f_{1,1}=1/c$, $f_{2,2}=1$, $f_{3,3}=1/c$, and $f_{4,4}=2/c^2$ or $4$ when $\Psi_4(c)=\log^2(c)$ or $c^2$, respectively. Apply Case $(c)$ of Theorem \ref{main1}, for any design, we can find a design with three points including end points $L$ and $U$, such that the off-diagonal elements are the same and diagonal elements are either the same or larger. Thus the conclusion follows.
\end{proof}
\begin{Remark} Han and Chaloner (2003) studied $D$- and $c$-optimal design under a slightly different exponential model where $\eta(x,\theta)=\theta_0+\theta_1\exp(-\theta_2 x)$. By applying the similar approach as used for the exponential model in Theorem \ref{example1}, we can show that for an arbitrary design $\xi$, there exists a design $\xi^*$ with three support points including low end point $L$ such that  $I_{\xi}(\theta)\leq I_{\xi^*}(\theta)$. This confirms and extends results in Han and Chaloner (2003) for this model, while  Han and Chaloner (2003) showed that  the $D$- and $c$-optimal designs are based on three points including two end points $L$ and $U$.
 \end{Remark}

Dette, Melas, and Wong (2005) studied another version of $E_{max}$ model, which can be written  in the form of  (\ref{model1}) with
\begin{equation}\begin{split}
\eta(x,\theta)=\frac{\theta_0x^{\theta_2}}{\theta_1+x^{\theta_2}}, \label{model3:2}
\end{split}\end{equation}
where $x\in[0,T]$, $\theta_0>0$, $\theta_1>0$, and $\theta_2\neq
0$. They showed that $D$- and $D_1$-optimal designs are based on
three points including end point $T$. The next theorem shows that
we can restrict ourself with three points designs for any optimal
designs.
\begin{Theorem}\label{example2}
Under model (\ref{model3:2}), for an arbitrary design $\xi$, there exists a design $\xi^*$ with three support points such that  $I_{\xi}(\theta)\leq I_{\xi^*}(\theta)$.
\end{Theorem}
\begin{proof}
By some routine algebra, it can be shown that the information matrix can be written in the form of (\ref{infor1}) with
\begin{equation}\begin{split}
P(\theta)=\begin{pmatrix}
1 & 0 & 0\\
1& \frac{\theta_1}{\theta_0} & 0\\
0& -\frac{\theta_1\log(\theta_1)}{\theta_0} & -\frac{\theta_2}{\theta_0}
\end{pmatrix}^{-1} \text{ and }
C(\theta,c_i)=\begin{pmatrix}
\frac{1}{(1+c_i)^2} &  \frac{1}{(1+c_i)^3}&   \frac{c_i\log(c_i)}{(1+c_i)^3}\\
\frac{1}{(1+c_i)^3}&  \frac{1}{(1+c_i)^4} &  \frac{c_i\log(c_i)}{(1+c_i)^4}\\
\frac{c_i\log(c_i)}{(1+c_i)^3}&  \frac{c_i\log(c_i)}{(1+c_i)^4}& \frac{c_i^2\log^2(c_i)}{(1+c_i)^4}
\end{pmatrix}, \label{example2:1}
\end{split}\end{equation}
where $c_i=\theta_1x_i^{-\theta_2}$.
Let $\Psi_1(c)=\frac{1}{(1+c)^4} $, $\Psi_2(c)=\frac{1}{(1+c)^3} $, $\Psi_3(c)=\frac{c\log(c)}{(1+c)^4}$, and $\Psi_4(c)=\frac{1}{(1+c)^2}$, $\Psi_5(c)=\frac{c\log(c)}{(1+c)^3}$, and $\Psi_6(c)=\frac{c^2\log^2(c)}{(1+c)^4}$. We can verify that the corresponding $f_{1,1}=-\frac{4}{(1+c)^5}$, $f_{2,2}=3/4$, $f_{3,3}=\frac{3c+1}{3c^2}$, $f_{4,4}=\frac{4c(3c+2)}{(3c+1)^2}$, $f_{5,5}=\frac{9c^3+15c^2+7c+1}{(3c+2)^2c^2}$, and $f_{6,6}=\frac{18c^2+15c+2}{c(9c^2+6c+1)}$. Notice that $c>0$, this implies that $F(c)<0$.  By Case (d) of Theorem \ref{main2}, the conclusion follows.
\end{proof}

\begin{Remark} Li and Majumdar (2008) studied $D$-optimal design for a three-parameter logistic model where $\eta(x,\theta)=\frac{\theta_0}{1+\theta_1\exp(\theta_2 x)}$. It can be shown that the information matrix can be written  in the form of  (\ref{example2:1}) with $c_i=\theta_1\exp(\theta_2 x_i)$. Thus for any arbitrary design $\xi$, there exists a design $\xi^*$ with three support points such that  $I_{\xi}(\theta)\leq I_{\xi^*}(\theta)$. This confirms and extends Li and Majumdar (2008)'s results.
 \end{Remark}

Han and Chaloner (2003) studied $D$- and $c$-optimal designs for a
model  which can be written  in the form of  (\ref{model1}) with
\begin{equation}\begin{split}
\eta(x,\theta)=\log(\theta_0+\theta_1\exp(-\theta_2 x)), \label{model3:3}
\end{split}\end{equation}
where $x\in[L,U]\subset (0,\infty)$, $\theta_0>0$, $\theta_1>0$,
and $\theta_2> 0$. They showed that $D$- and $c$-optimal designs
are based on three points including end points $L$ and $U$. In
fact, any optimal design based on information matrix can be
restricted to a three-point design including lower end point $L$.
\begin{Theorem}\label{example3}
Under (\ref{model3:3}), for any arbitrary design $\xi$, there exists a design $\xi^*$ with three support points including lower end point $L$ such that  $I_{\xi}(\theta)\leq I_{\xi^*}(\theta)$.
\end{Theorem}
\begin{proof}
It can be shown that the information matrix can be written in the form of (\ref{infor1}) with
\begin{equation}\begin{split}
P(\theta)=\begin{pmatrix}
\theta_0 & \theta_1 & 0\\
\theta_0&0 & 0\\
0& \theta_1\log(\frac{\theta_1}{\theta_0}) & \theta_2
\end{pmatrix}^{-1} \text{ and }
C(\theta,c_i)=\begin{pmatrix}
1 &  \frac{1}{1+c_i}&   \frac{c_i\log(c_i)}{1+c_i}\\
\frac{1}{1+c_i}&  \frac{1}{(1+c_i)^2} &  \frac{c_i\log(c_i)}{(1+c_i)^2}\\
 \frac{c_i\log(c_i)}{1+c_i}&  \frac{c_i\log(c_i)}{(1+c_i)^2}& \frac{c_i^2\log^2(c_i)}{(1+c_i)^2}
\end{pmatrix}, \label{example3:1}
\end{split}\end{equation}
where $c_i=\frac{\theta_1}{\theta_0}\exp(-\theta_2 x_i)$. Let
$\Psi_1(c)=\frac{1}{(1+c)^2} $, $\Psi_2(c)=\frac{1}{1+c} $,
$\Psi_3(c)=\frac{c\log(c)}{(1+c)^2}$, and
$\Psi_4(c)=\frac{c\log(c)}{1+c}$, and
$\Psi_5(c)=\frac{c^2\log^2(c)}{(1+c)^2}$. We can verify that the
corresponding $f_{1,1}=-\frac{2}{(1+c)^3}$, $f_{2,2}=\frac{1}{2}$,
$f_{3,3}=\frac{1+c}{c^2}$, $f_{4,4}=-2$, and
$f_{5,5}=\frac{2}{c}$. Notice that $c>0$,  which implies that
$F(c)>0$.  By Case (b) of Theorem \ref{main2} and
$c_i=\frac{\theta_1}{\theta_0}\exp(-\theta_2 x_i)$, the conclusion
follows.
\end{proof}

\subsection{Models with four parameters}
Fang and Hedayat (2008) studied a composed $E_{max}$-PK1 model, which  can be written  in the form of  (\ref{model1}) with
\begin{equation}\begin{split}
\eta(x,\theta)=\theta_0+\frac{\theta_1D}{D+\theta_2\exp(\theta_3 x)}. \label{model4:2}
\end{split}\end{equation}
Here, $D$ is a positive constant, $x_i\in [0,U]$, and $\theta_i>0$, $i=0,1,2,3$. Fang and Hedayat (2008) showed that local $D$-optimal designs are based on four points including end points $0$ and $U$. The next theorem tells us that any optimal designs can be based on four points designs.
\begin{Theorem}\label{example5}
Under model (\ref{model4:2}), for any arbitrary design $\xi$,
there exists a design $\xi^*$ with four support points such that
$I_{\xi}(\theta)\leq I_{\xi^*}(\theta)$.
\end{Theorem}
\begin{proof}
It can be shown that the information matrix can be written  in the form of
(\ref{infor1}) with
\begin{equation}\begin{split}
P(\theta)&=\begin{pmatrix}
1 & 0 & 0 & 0\\
0& 1 & 0 & 0\\
0& 1& \frac{\theta_2}{\theta_1} & 0\\
0 & 0& -\frac{\theta_2}{\theta_1} \log(\frac{\theta_2}{D})&  -\frac{\theta_3}{\theta_1}
\end{pmatrix}^{-1} \\
\text{ and } C(\theta,c_i)&=\begin{pmatrix}
1 &   \frac{1}{1+c_i} &   \frac{1}{(1+c_i)^2} & \frac{c_i \log(c_i)}{(1+c_i)^2}\\
\frac{1}{1+c_i} & \frac{1}{(1+c_i)^2} &  \frac{1}{(1+c_i)^3} &  \frac{c_i \log(c_i)}{(1+c_i)^3}\\
 \frac{1}{(1+c_i)^2}& \frac{1}{(1+c_i)^3}& \frac{1}{(1+c_i)^4} &  \frac{c_i \log(c_i)}{(1+c_i)^4}\\
 \frac{c_i \log(c_i)}{(1+c_i)^2} &  \frac{c_i \log(c_i)}{(1+c_i)^3} &  \frac{c_i \log(c_i)}{(1+c_i)^4} &  \frac{c_i^2 \log^2(c_i)}{(1+c_i)^4}
  \end{pmatrix},\label{example4:2}
\end{split}\end{equation}
where $c_i=\frac{\theta_2}{D}\exp(\theta_3x_i)$. Let
$\Psi_1(c)=\frac{1}{(1+c)^4} $, $\Psi_2(c)=\frac{1}{(1+c)^3} $,
$\Psi_3(c)=\frac{c\log(c)}{(1+c)^4}$, and
$\Psi_4(c)=\frac{1}{(1+c)^2} $,
$\Psi_5(c)=\frac{c\log(c)}{(1+c)^3}$, $\Psi_6(c)=\frac{1}{1+c}$,
$\Psi_7(c)=\frac{c\log(c)}{(1+c)^2}$, and
$\Psi_8(c)=\frac{c^2\log^2(c)}{(1+c)^4}$. We can verify that the
corresponding $f_{1,1}=-\frac{4}{(1+c)^5}$, $f_{2,2}=\frac{3}{4}$,
$f_{3,3}=\frac{3c+1}{3c^2}$, $f_{4,4}=\frac{4c(3c+2)}{(3c+1)^2}$,
$f_{5,5}=\frac{9c^3+15c^2+7c+1}{c^2(3c+2)^2}$,
$f_{6,6}=\frac{9c(3c+2)}{9c^2+6c+1}$, $f_{7,7}=\frac{3c+1}{3c^2}$,
and $f_{8,8}=\frac{2}{3c^2}$. Notice that $c>0$, which implies
that $F(c)<0$.  By applying  Case (d) of Theorem \ref{main2}, the
conclusion follows.
\end{proof}

Dette, Bretz, Pepelyshev, and Pinheiro (2008) studied a
four-parameter logistic model, which  can be written  in the form of
(\ref{model1}) with
\begin{equation}\begin{split}
\eta(x,\theta)=\theta_0+\frac{\theta_1}{1+\exp(\frac{\theta_2-x}{\theta_3})}. \label{model4:1}
\end{split}\end{equation}
Here, $x_i\in [L,U]\subset (0,\infty)$, $\theta_1>0$,
$\theta_2>0$, and $\theta_3>0$. Although they did not provide an
analytical solution,  their numerical solution shows that local
MED-optimal designs are based on four points including the end point
$L$. In fact, any optimal design under  (\ref{model4:1}) can be
based on a four-point design. Notice that the information matrix
of Model (\ref{model4:1}) can be written as (\ref{example4:2})
except that $c_i=\exp(\frac{\theta_2-x_i}{\theta_3})$ and
\begin{equation}\begin{split}
P(\theta)=\begin{pmatrix}
1 & 0 & 0 & 0\\
0& 1 & 0 & 0\\
0& 1& \frac{\theta_3}{\theta_1} & 0\\
0 & 0& 0&  \frac{\theta_3}{\theta_1}
\end{pmatrix}^{-1}.
\end{split}\end{equation}
Immediately, we have the following theorem.
\begin{Theorem}\label{example4}
Under model (\ref{model4:1}), for any arbitrary design $\xi$,
there exists a design $\xi^*$ with four support points such that
$I_{\xi}(\theta)\leq I_{\xi^*}(\theta)$.
\end{Theorem}
Theorem \ref{example4} confirms and extends Dette, Bretz, Pepelyshev, and Pinheiro (2008)'s numerical results.

\begin{Remark} Li and Majumdar (2008) studied $D$-optimal design for a different version of four-parameter logistic model where $$\eta(x,\theta)=\theta_0+\frac{\theta_1}{1+\exp(\theta_2+\theta_3 x)}.$$ It can be shown that the information matrix can be written in the form of (\ref{example4:2}) with $c_i=\exp(\theta_2+\theta_3 x_i)$ and
\begin{equation}\begin{split}
P(\theta)=\begin{pmatrix}
1 & 0 & 0 & 0\\
0& 1 & 0 & 0\\
0& 1& \frac{1}{\theta_1} & 0\\
0 & 0& -\frac{\theta_2}{\theta_1}&  -\frac{\theta_3}{\theta_1}
\end{pmatrix}^{-1}.
\end{split}\end{equation}
 Thus for an arbitrary design $\xi$, there exists a design $\xi^*$ with at most four support points such that  $I_{\xi}(\theta)\leq I_{\xi^*}(\theta)$. This confirms and extends the results in Li and Majumdar (2008) for this model. \end{Remark}

\subsection{Models with $p+1$ parameters }\label{nparm}
Therorem \ref{main2} can be applied to many classical polynomial
regression models. de la Garza  (1954) studied a $p$th-degree
polynomial regression model, which  can be written  in the form of
(\ref{model1}) with
\begin{equation}\begin{split}
\eta(x,\theta)=\theta_0+\sum_{i=1}^p\theta_i x^i, \label{model5:1}
\end{split}\end{equation}
where $x\in[-1,1]$. de la Garza  (1954) proved that any optimal
design can be based on at most $p+1$ points including end points
$-1$ and $1$. Here, we provide an alternative way to prove this
result under a general design region $[L,U]$.
\begin{Theorem} [de la Garza Phenomenon] \label{example6}
Under model (\ref{model5:1}), for any arbitrary design $\xi$,
there exists a design $\xi^*$ with $p+1$ support points including
end points $L$ and $U$ such that  $I_{\xi}(\theta)\leq
I_{\xi^*}(\theta)$.
\end{Theorem}
\begin{proof}
The information matrix can be written in the form of (\ref{infor1}) with $P(\theta)=I_{(p+1)\times (p+1)}$ and
\begin{equation}\begin{split}
C(\theta,c_i)=\begin{pmatrix}
1 &   c_i &  \ldots &c_i^p\\
c_i &c_i^2 &  \ldots &c_i^{p+1}\\
 \vdots &\vdots &\ddots &\vdots \\
c_i^p & c_i^{p+1} &\ldots &c_i^{2p}
  \end{pmatrix},\label{example5:2}
\end{split}\end{equation}
where $c_i=x_i$. Let $\Psi_l(c)=c^l $, $l=1,\ldots,2p$. We can
check that the corresponding $f_{l,l}=l$ for $l=1,\ldots,2p$. By
applying  Case (c) of Theorem \ref{main2}, we can draw the desired
conclusion.
\end{proof}

Weighted polynomial regression is an extension of Model (\ref{model5:1}), where the error terms $\epsilon_{ij}$'s are i.i.d $N(0,\sigma^2/\lambda(x))$ ($\sigma^2$ is known).
Both Karlin and Studden (1966) and Dette, Haines, and Imhof (1999) studied $D$-optimal designs under various choices of $\lambda(x)$ and design regions. Their results show that the number of support points of $D$-optimal designs is $p+1$ (except Lemma 2.2 of  Dette, Haines, and Imhof, 1999, which has at most $p+2$ points). By applying Therorem \ref{main2},  we can extend their conclusions to any  optimal designs.  The results are summarized below:
\begin{Theorem}  \label{example7}
Under model (\ref{model5:1}),  where the error terms $\epsilon_{ij}$'s are i.i.d $N(0,\sigma^2/\lambda(x))$ ($\sigma^2$ is known), for an arbitrary design $\xi$, there exists a design $\xi^*$  such that  $I_{\xi}(\theta)\leq I_{\xi^*}(\theta)$. Here, $\xi^*$ is defined as follows:
\begin{itemize}
\item[(i)] $\xi^*$ is based on at most $p+1$ points when $\lambda(x)=(1-x)^{\alpha+1}(1+x)^{\beta+1}$, $x\in [-1,1]$, $\alpha+1>0$, and $\beta+1>0$.
\item[(ii)] $\xi^*$ is based on at most  $p+1$ points including point 0 when $\lambda(x)=\exp(-x)$ and $x\geq 0$.
\item[(iii)] $\xi^*$ is based on  at most $p+1$ points  when $\lambda(x)=x^{\alpha+1}\exp(-x)$, $x\geq 0$, and $\alpha+1>0$.
\item[(iv)] $\xi^*$ is based on  at most $p+1$ points  when $\lambda(x)=\exp(-x^2)$.
\item[(v)] $\xi^*$ is based on  at most $p+1$ points  when $\lambda(x)=(1+x^2)^{-n}$ and $p\leq n$.
\item[(vi)] $\xi^*$ is based on at most $p+2$ points including either lower end point $L$ or upper end point $U$ when $\lambda(x)=(1+x^2)^{-n}$ and $p>n$.
\end{itemize}
\end{Theorem}
\begin{proof}
The information matrix can be written  in the form of  (\ref{infor1}) with
$P(\theta)=I_{(p+1)\times (p+1)}$ and
\begin{equation}\begin{split}
C(\theta,c_i)=\begin{pmatrix}
\lambda(c_i) &   \lambda(c_i)c_i &  \ldots &\lambda(c_i)c_i^p\\
\lambda(c_i)c_i &\lambda(c_i)c_i^2 &  \ldots &\lambda(c_i)c_i^{p+1}\\
 \vdots &\vdots &\ddots &\vdots \\
\lambda(c_i)c_i^p & \lambda(c_i)c_i^{p+1} &\ldots &\lambda(c_i)c_i^{2p}
  \end{pmatrix},\label{example6:2}
\end{split}\end{equation}
where $c_i=x_i$. The proofs are similar for all cases except for
(ii), which can be proven with a similar approach as in Theorem
\ref{example6} proof. Here we give the proof of Theorem \ref{example7} (i). Let
$\Psi_1(c)=-\int_{0}^c(1-t)^{\alpha}(1+t)^{\beta}dt $ and
$\Psi_l(c)=(1-c)^{\alpha+1}(1+c)^{\beta+1}c^{l-2} $,
$l=2,\ldots,2p+2$. Notice that $\Psi_1(c)$ is not one of the elements
in (\ref{example6:2}). We simply choose its value here for computation
convenience. We can check that the corresponding $f_{1,1}<0$ and
$f_{l,l}>0$ for $l=2,\ldots,2p+2$. By applying Case (d) of Theorem
\ref{main2}, the conclusion follows.
\end{proof}

\subsection{Loglinear model with quadratic term} \label{loglinear}
Theorem \ref{main2} is not limited to the model format
(\ref{model1}). It can be applied to any nonlinear model, as long
as the information matrix can be written in the form of (\ref{infor1}). Here
we give one such example. Wang, Myers, Smith and Ye (2006) studied
$D$-optimal designs for loglinear models with a quadratic term,
where
\begin{equation}\begin{split}
y_i\sim \text{ Poisson}(\mu_i) \text{ and } \log(\mu_i)=\theta_0+\theta_1 x_i+\theta_2 x_i^2. \label{loglinearmodel}
\end{split}\end{equation}
Here, $x_i\in [L,U]$. They showed that $D$-optimal designs are
based on three points for some selected parameters by numerical
searching. Their conclusion can be verified with the following
theorem.
\begin{Theorem}\label{example8}
Under model (\ref{loglinearmodel}), for any arbitrary design
$\xi$, there exists a design $\xi^*$  such that
$I_{\xi}(\theta)\leq I_{\xi^*}(\theta)$. Here, when $\theta_2<0$,
$\xi^*$ is based on three points; when $\theta_2>0$, $\xi^*$ is
based on four points including  the end points $L$ and $U$.
\end{Theorem}
\begin{proof}
It can be shown that the information matrix can be written  in the form of
(\ref{infor1}) with
\begin{equation}\begin{split}
P(\theta)&=\exp(\frac{4\theta_2\theta_0-\theta_1^2}{8\theta_2})\begin{pmatrix}
1 & 0 & 0\\
\frac{\theta_1 \sqrt{|\theta_2|}}{2\theta_2}& \sqrt{|\theta_2|} & 0 \\
\text{sign}(\theta_2)\frac{\theta_1^2}{4\theta_2}& \text{sign}(\theta_2)\theta_1& \text{sign}(\theta_2)\theta_2
\end{pmatrix}^{-1} \\
\text{ and }C(\theta,c_i)&=\begin{pmatrix}
 e^{\text{sign}(\theta_2)c_i^2}&   c_i e^{\text{sign}(\theta_2)c_i^2}&  c_i^2e^{\text{sign}(\theta_2)c_i^2}\\
c_i e^{\text{sign}(\theta_2)c_i^2} & c_i^2 e^{\text{sign}(\theta_2)c_i^2} & c_i^3 e^{\text{sign}(\theta_2)c_i^2}\\
c_i^2 e^{\text{sign}(\theta_2)c_i^2} &  c_i^3 e^{\text{sign}(\theta_2)c_i^2} & c_i^4 e^{\text{sign}(\theta_2)c_i^2}  \end{pmatrix},\label{loglinearc}
\end{split}\end{equation}
where
$c_i=\sqrt{|\theta_2|}x_i+\frac{\theta_1\sqrt{|\theta_2|}}{2\theta_2}$.
Let $\Psi_1(c)=\text{sign}(\theta_2)\int_{0}^c
e^{\text{sign}(\theta_2)t^2} dt $ and
$\Psi_l(c)=c^{l-2}e^{\text{sign}(\theta_2)c^2} $, $l=2,\ldots,6$.
We can verify that the corresponding (i) $f_{1,1}<0$ when
$\theta_2<0$ or $f_{1,1}>0$ when $\theta_2>0$; (ii) $f_{l,l}>0$
for $l=2,\ldots,6$. By Theorem \ref{main2}, the conclusion
follows.
\end{proof}

\begin{Remark}
Notice that the first derivative of $e^{\text{sign}(\theta_2)c^2}$
is 0 when $c=0$. So we cannot apply Theorem 2 if $c$ ranges from
a negative to a positive number and $\Psi_1(c)=e^{\text{sign}(\theta_2)c^2}$.
To avoid this situation, a specific $\Psi_1(c)$ is chosen
although it is not among the functions in (\ref{loglinearc}). This
is the general strategy to handle such situations. The
disadvantage of this strategy is that it could increase the number
of support points unnecessarily.
\end{Remark}

\section{Discussion}
Deriving optimal designs for nonlinear models is complicated.
Currently, the main tools are Elfving's geometric approach and
Kiefer's equivalence theorem. Although these two approaches have
been proven to be powerful tools, the results have to be derived
on a case-by-case basis and some optimal designs are difficult to
derive. In contrast, the proposed approach in this paper can yield
very general results. As we have illustrated in the last section,
for many commonly studied nonlinear models, this approach gives
some simple structures based on which any optimal design can be
found. As a result, it is a relatively easy to find an optimal
design since one only needs to consider these simple structures.
Many practical models have a moderate number of parameters (see
Sections 4.1 and 4.2). For those models, any optimal designs can
be derived readily. At a minimum, numerical search is feasible
with the algorithm proposed by Stufken and Yang (2009).

The well-known  Carath$\acute{\text{e}}$odory's theorem gives
$p(p+1)/2$ as a upper bound for the number of support points in
optimal designs. Examples in Section 4 show that the upper bound
can be as small as $p$, the minimum number of support points such
that all parameters are estimable. Although this may not be true
for arbitrary nonlinear models, the proposed approach can be used
to improve the upper bound. On the other hand, this approach gives
an alternative way to prove  the de la Garza Phenomenon with
little effort.  Furthermore, this phenomenon is extended for more
general weighted polynomial regression models.

The proposed approach offers a lot of flexibility. It can be
applied to multi-stage design, an important feature for locally
optimal design. It works for any design region. The conditions are
mild and can be easily verified using symbolic computational
software packages, such as Maple or Mathematica.

While the results of this paper are already far reaching, we
believe that there is potential  to extend it further. In general,
this approach can be applied to any nonlinear model as long as the
corresponding functions are differentiable. One possible obstacle
is that some $f_{l,l}$ may take the value 0. In this situation,
the proposed approach may not be applied directly. One way to
handle it is to introduce some new functions. For example, refer
to the proof of Theorem \ref{example8}. This, however, may
increase the number of support points unnecessarily. How to handle
this situation remains an open question for future research.

\section{Appendix}
The purpose of this appendix is to present the proof of Lemma
\ref{exist}. Most of this is done by mathematical induction. We
first assume that Lemma \ref{exist} holds for $k\leq K$, then we
will show that the lemma also holds for $K+1$. When $k=2$ and $3$,
Lemma \ref{exist} has been proven by Yang and Stufken (2009)
(Propositions A.2 and A.3 for $k=2$; Lemmas 2 and 3   for $k=3$).
There are two sections in this appendix. In Section 6.1, we
present some useful propositions in preparation for the main
proof. In the Section 6.2, we provide the main proof.
\subsection{Some useful propositions}

\begin{Proposition}\label{prop2}
Let $\Psi_1, \ldots, \Psi_k$ be $k$ functions defined on $[A,B]$.
Assume that $f_{l,l}(c)>0$, $c\in[A,B]$, $l=1,\ldots,k$. Then for
any $A\leq c_1<\ldots <c_m \leq B$,
$g(\Psi'_1,\ldots,\Psi'_m,c_1,\ldots,c_m)>0$, where $m\leq k$ and function $g$
is defined as
\begin{equation}\begin{split}
g(\Psi'_1,\ldots,\Psi'_m,c_1,\ldots,c_m)=\left|
\begin{array}{cccc}
  \Psi'_1(c_1) & \Psi'_1(c_2) & \ldots & \Psi'_1(c_m) \\
  \Psi'_2(c_1) & \Psi'_2(c_2) & \ldots & \Psi'_2(c_m) \\
  \vdots &  \vdots & \ddots & \vdots \\
  \Psi'_m(c_1) &  \Psi'_m(c_2) & \ldots &  \Psi'_m(c_m)
\end{array} \right|. \label{function:g}
\end{split}\end{equation}
\end{Proposition}
\begin{proof}
By the definition of $f_{l,t}(c)$ in (\ref{def:df}) and  $f_{1,1}(c)>0$ for all $c\in[A,B]$, it is sufficient to show that
\begin{equation}\begin{split}
\left|
\begin{array}{ccccc}
  1 & 1 & \ldots & 1 & 1\\
\frac{f_{21}}{f_{11}}(c_1) & \frac{f_{21}}{f_{11}}(c_2)  & \ldots & \frac{f_{21}}{f_{11}}(c_{m-1}) & \frac{f_{21}}{f_{11}}(y_m) \\
  \vdots &  \vdots & \ddots & \vdots & \vdots \\
\frac{f_{m1}}{f_{11}}(c_1) &\frac{f_{m1}}{f_{11}}(c_2) & \ldots & \frac{f_{m1}}{f_{11}}(c_{m-1})& \frac{f_{m1}}{f_{11}}(y_m)
\end{array} \right|_{y_m=c_m}>0. \label{prop2:1}
\end{split}\end{equation}
Notice that if we treat $y_m$ as a variable ranging from $c_{m-1}$
to $B$, the determinant of the left hand side of (\ref{prop2:1})
is 0 when $y_m=c_{m-1}$. So it suffices to show that its
derivative about $y_m$ is positive when $y_m>c_{m-1}$, i.e.,
\begin{equation}\begin{split}
\left|
\begin{array}{ccccc}
  1 & 1 & \ldots & 1& 0 \\
\frac{f_{21}}{f_{11}}(c_1) & \frac{f_{21}}{f_{11}}(c_2)  & \ldots & \frac{f_{21}}{f_{11}}(y_{m-1}) & f_{22}(y_m) \\
  \vdots &  \vdots & \ddots & \vdots & \vdots \\
\frac{f_{m1}}{f_{11}}(c_1) &\frac{f_{m1}}{f_{11}}(c_2) & \ldots & \frac{f_{m1}}{f_{11}}(y_{m-1})& f_{m2}(y_m)
\end{array} \right|_{y_{m-1}=c_{m-1}}>0. \label{prop2:2}
\end{split}\end{equation}
For (\ref{prop2:2}), apply the same argument for $y_m$ to
$y_{m-1}$, the conclusion follows if we can show that for
$y_{m-1}\in (c_{m-2}, y_m)$,
\begin{equation*}\begin{split}
\left|
\begin{array}{cccccc}
  1 & 1 & \ldots &1 &0 & 0 \\
\frac{f_{21}}{f_{11}}(c_1)& \frac{f_{21}}{f_{11}}(c_2) & \ldots & \frac{f_{21}}{f_{11}}(y_{m-2})& f_{22}(y_{m-1})& f_{22}(y_m) \\
  \vdots &  \vdots & \ddots & \vdots& \vdots  & \vdots \\
\frac{f_{m1}}{f_{11}}(c_1)& \frac{f_{m1}}{f_{11}}(c_2) & \ldots & \frac{f_{m1}}{f_{11}}(y_{m-2})& f_{m2}(y_{m-1}) &
 f_{m2}(y_m)
\end{array} \right|_{y_{m-2}=c_{m-2}}>0.
\end{split}\end{equation*}
Repeat the exact same arguments for $y_{m-2}\in [c_{m-3},
y_{m-1}]$ and so on, until $y_2 \in [c_1,y_3]$, then it is sufficient to show that
\begin{equation}\begin{split}
\left|
\begin{array}{cccc}
  f_{22}(y_2) & f_{22}(y_3) & \ldots & f_{22}(y_m) \\
  f_{32}(y_2) & f_{32}(y_3) & \ldots & f_{32}(y_m) \\
  \vdots &  \vdots & \ddots & \vdots \\
   f_{m2}(y_2) & f_{m2}(y_3) & \ldots & f_{m2}(y_m)
\end{array} \right|>0 \label{prop2:3}
\end{split}\end{equation}
for any $A\leq y_2<y_3<\ldots<y_m\leq B$. Notice that $f_{22}(c)>0$ for all $c\in [A,B]$, repeat the same argument as for $f_{1,1}(c)$, it is sufficient to show that
\begin{equation}\begin{split}
\left|
\begin{array}{cccc}
  f_{33}(z_3) & f_{33}(z_4) & \ldots & f_{33}(z_m) \\
  f_{43}(z_3) & f_{43}(z_4) & \ldots & f_{43}(z_m) \\
  \vdots &  \vdots & \ddots & \vdots \\
   f_{m3}(z_3) & f_{m3}(z_4) & \ldots & f_{m3}(z_m)
\end{array} \right|>0. \label{prop2:4}
\end{split}\end{equation}
for any $A\leq z_3<z_4<\ldots<z_m\leq B$. Repeat the same argument for $f_{33}$ and so on until $f_{m-1,m-1}$, it is sufficient to show that $f_{mm}(c)>0$ for any $c\in [A,B]$, which is one of our assumptions.  Thus the conclusion follows.
\end{proof}

\begin{Proposition}\label{prop3}
Let $\Psi_1, \ldots,
\Psi_k$ be $k$ functions defined on $[A,B]$. Assume that $f_{l,l}(c)>0$, $c\in[A,B]$, $l=1,\ldots,k$.   Then for  $A\leq a_1<b_1\leq a_2 <b_2 \leq \ldots \leq a_m<b_m \leq B$, where $m\leq k$, we have
\begin{equation}\begin{split}
D(\Psi(b_1)-\Psi(a_1), \ldots, \Psi(b_m)-\Psi(a_m))>0. \label{prop3:0}
\end{split}\end{equation}
Here,
\begin{equation}\begin{split}
D(\Psi(b_1)-&\Psi(a_1), \ldots, \Psi(b_m)-\Psi(a_m))=\\
&\left|
\begin{array}{cccc}
 \Psi_1(b_1)-\Psi_1(a_1) &  \Psi_1(b_2)-\Psi_1(a_2) & \ldots &  \Psi_1(b_m)-\Psi_1(a_m) \\
  \Psi_2(b_1)-\Psi_2(a_1) &  \Psi_2(b_2)-\Psi_2(a_2) & \ldots &  \Psi_2(b_m)-\Psi_2(a_m)\\
  \vdots &  \vdots & \ddots & \vdots \\
 \Psi_m(b_1)-\Psi_m(a_1) &  \Psi_m(b_2)-\Psi_m(a_2) & \ldots &  \Psi_m(b_m)-\Psi_m(a_m)
\end{array}\right|. \label{prop3:1}
\end{split}\end{equation}
\end{Proposition}
\begin{proof}
Consider the left hand side of (\ref{prop3:1}) as a function of $b_m$.  When $b_m=a_m$, it is clear that
it is 0. The conclusion is sufficient if we can show that
 \begin{equation}\begin{split}
\left|
\begin{array}{cccc}
 \Psi_1(b_1)-\Psi_1(a_1)  & \ldots &  \Psi_1(b_{m-1})-\Psi_1(a_{m-1}) &   \Psi'_1(c_m) \\
  \Psi_2(b_1)-\Psi_2(a_1)  & \ldots &  \Psi_2(b_{m-1})-\Psi_2(a_{m-1})& \Psi'_2(c_m)\\
  \vdots &  \ddots & \vdots & \vdots \\
 \Psi_m(b_1)-\Psi_m(a_1) & \ldots &  \Psi_m(b_{m-1})-\Psi_m(a_{m-1})  & \Psi'_m(c_m)
\end{array} \right|>0 \label{prop3:2}
\end{split}\end{equation}
for $c_m\in (a_m,B]$. On the other hand, the determinant in (\ref{prop3:2}) is zero when $b_{m-1}=a_{m-1}$. So the conclusion is sufficient if we can show that its derivative respective to $b_{m-1}$ is positive for $b_{m-1}\in (a_{m-1},c_m]$. Repeat this argument for $b_{m-2}$, $b_{m-3}$ and so on until $b_1$. Then the conclusion is sufficient if
\begin{equation}\begin{split}
\left|
\begin{array}{cccc}
  \Psi'_1(c_1) & \Psi'_1(c_2) & \ldots & \Psi'_1(c_m) \\
  \Psi'_2(c_1) & \Psi'_2(c_2) & \ldots & \Psi'_2(c_m) \\
  \vdots &  \vdots & \ddots & \vdots \\
  \Psi'_m(c_1) &  \Psi'_m(c_2) & \ldots &  \Psi'_m(c_m)
\end{array} \right|>0 \label{prop3:3}
\end{split}\end{equation}
for $A\leq c_1<c_2<\ldots<c_m\leq B$. By Proposition \ref{prop2}, we have our conclusion.
\end{proof}

\begin{Corollary}
Let $\Psi_1, \ldots,
\Psi_k$ be $k$ functions defined on $[A,B]$. Assume that $f_{l,l}(c)>0$, $c\in[A,B]$, $l=1,\ldots,k$. Let $S_{m+1}=\{s_1,\ldots,s_{m+1}\}$, where $m\leq k$, $s_i\in [A,B], i=1,\ldots,m+1$ and $s_i<s_{i+1}, i=1,\ldots,m$.
Define
\begin{equation}\begin{split}
D_1(\Psi,
S_{m+1})=D(\Psi(s_2)-\Psi(s_1),\Psi(s_3)-\Psi(s_2),\ldots,\Psi(s_{m+1})-\Psi(s_{m})),
\label{function:D1}
\end{split}\end{equation}
where function $D$
is defined as (\ref{prop3:1}). Then
\begin{equation}\begin{split}
D_1(\Psi,
S_{m+1})>0.
\end{split}\end{equation}
\end{Corollary}

\begin{Proposition}\label{even}
Let $\Psi_1, \ldots,
\Psi_k$ be $k=2n$ functions defined on $[A,B]$. Assume that $f_{l,l}(c)>0$, $c\in[A,B]$, $l=1,\ldots,k$. Suppose for some $c_1,\ldots,c_n$ and $\tilde{c}_0,\ldots,\tilde{c}_{n}$, where
$A\leq \tilde{c}_0<c_1<\tilde{c}_1<c_2<\tilde{c}_2<\ldots<c_n<\tilde{c}_n\leq B$, there exist
$\omega_1,\ldots,\omega_n$ and $\tilde{\omega}_0,\ldots,\tilde{\omega}_n$, such that $\sum_{i=1}^n
\omega_i=\sum_{j=0}^n \tilde{\omega}_j$ and the following $k-1$ equations hold:
\begin{equation}\begin{split}
\sum_{i=1}^{n}\omega_i\Psi_l(c_{i}) =\sum_{j=0}^{n}\tilde{\omega}_j\Psi_l(\tilde{c}_{j}),
l=1,\ldots,k-1. \label{even:1}
\end{split}\end{equation}
If at least one of $\omega_1,\ldots,\omega_n$ and $\tilde{\omega}_0,\ldots,\tilde{\omega}_n$ is positive, then all of them should be positive.
Under this situation,
\begin{equation}\begin{split}
\sum_{i=1}^{n}\omega_i\Psi_k(c_{i}) <\sum_{j=0}^{n}\tilde{\omega}_j\Psi_k(\tilde{c}_{j}). \label{even:2}
\end{split}\end{equation}
\end{Proposition}

\begin{proof}
For ease of presentation,  we define
$S_{k}(\overline{\tilde{c}_l})=\{c_1,\ldots,c_n,\tilde{c}_0,\tilde{c}_1,\ldots,\tilde{c}_{n}\}-\{\tilde{c}_l\}$
(all points but point $\tilde{c}_l$). Define $r=\sum_{i=1}^n
\omega_i$, then we have $\omega_n=r-\sum_{i=1}^{n-1}\omega_i$ and
$\tilde{\omega}_n=r-\sum_{j=0}^{n-1}\tilde{\omega}_j$. We can
treat (\ref{even:1}) as a system of linear equations
$\omega_1,\ldots, \omega_{n-1}$ and
$\tilde{\omega}_0,\ldots,\tilde{\omega}_{n-1}$. Let
$\Gamma=(\tilde{\omega}_0,\omega_1,\tilde{\omega}_1,\ldots,\omega_{n-1},\tilde{\omega}_{n-1})'$;
$(b_{2i+1},a_{2i+1})=(\tilde{c}_n,\tilde{c}_{i})$,
$i=0,\ldots,n-1$; and  $(b_{2i},a_{2i})=(c_i,c_n)$,
$i=1,\ldots,n-1$, then
\begin{equation}\begin{split}
\Gamma=
r\begin{pmatrix}
\Psi_1(b_1)-\Psi_1(a_1) & \ldots &  \Psi_1(b_{k-1})-\Psi_1(a_{k-1}) \\
  \vdots  & \ddots & \vdots \\
 \Psi_{k-1}(b_1)-\Psi_{k-1}(a_1)  & \ldots &  \Psi_{k-1}(b_{k-1})-\Psi_{k-1}(a_{k-1})
 \end{pmatrix}^{-1}
 \begin{pmatrix}
\Psi_1(\tilde{c}_n)-\Psi_1(c_n)\\
  \vdots \\
\Psi_{k-1}(\tilde{c}_n)-\Psi_{k-1}(c_n)
 \end{pmatrix} \label{even:3}
 \end{split}\end{equation}
For each $\tilde{\omega}_j$, $j=0,\ldots,n-1$, from (\ref{even:3}), by basic
properties of algebra, we have
\begin{equation}\begin{split}
\tilde{\omega}_j=r\frac{D(\Psi(\tilde{b}_1)-\Psi(\tilde{a}_1),\ldots,
\Psi(\tilde{b}_{k-1})-\Psi(\tilde{a}_{k-1}))}{D(\Psi(b_1)-\Psi(a_1),\ldots,\Psi(b_{k-1})-\Psi(a_{k-1}))}, \label{even:q}
 \end{split}\end{equation}
 where $(\tilde{b}_i,\tilde{a}_i)=(b_i,a_i)$ for all $i\neq 2j+1$ and $(\tilde{b}_{2j+1},\tilde{a}_{2j+1})=(\tilde{c}_n,c_n)$.
Function $D$
is defined as (\ref{prop3:1}).

We first study the property of
$D(\Psi(\tilde{b}_1)-\Psi(\tilde{a}_1),\ldots,\Psi(\tilde{b}_{k-1})-\Psi(\tilde{a}_{k-1}))$.
Let $C_i$ be the $i$th column of the corresponding matrix. Let
$\tilde{C}_i=(-1)^{i-1}[C_i+C_{i+1}-C_{2j+1}]$, $i=1,\ldots,2j-1$,
$2j+2,\ldots,k-2$,  $\tilde{C}_{k-1}=C_{k-1}-C_{2j+1}$,
$\tilde{C}_{2j}=-C_{2j}+C_{2j+2}$, and
$\tilde{C}_{2j+1}=C_{2j+1}$. Consider the new matrix
$$(\tilde{C}_1|\ldots|\tilde{C}_{2j}|\tilde{C}_{2j+2}|\ldots|\tilde{C}_{k-1}|\tilde{C}_{2j+1}).$$
Then we have
\begin{equation}\begin{split}
D(\Psi(\tilde{b}_1)-\Psi(\tilde{a}_1),\ldots,\Psi(\tilde{b}_{k-1})-\Psi(\tilde{a}_{k-1}))=(-1)^{n-1}
D_1(\Psi, S_{k}(\overline{\tilde{c}_j})), \label{even:4}
\end{split}\end{equation}
where $D_1$ is as defined in (\ref{function:D1}). The equality in
(\ref{even:4}) holds by applying some basic properties of a
determinant and considering the fact that $k$ is even.

Next, we study
$D(\Psi(b_1)-\Psi(a_1),\ldots,\Psi(b_{k-1})-\Psi(a_{k-1}))$. Let
$C_i$ be the $i$th column of the corresponding matrix. Let
$\tilde{C}_i=C_i-C_{i+2}$ when $i$ is odd ($i <k-1$) and
$\tilde{C}_{k-1}=C_{k-1}$; $\tilde{C}_i=-C_i+C_{i+2}$ when $i$ is
even ($i< k-2$) and $\tilde{C}_{k-2}=C_{k-2}$. Consider the new
matrix
$$\mathcal{M}=(\tilde{C}_1|\ldots|\tilde{C}_{k-1}).$$
Then we have
\begin{equation}\begin{split}
D(\Psi(b_1)-\Psi(a_1),\ldots,\Psi(b_{k-1})-\Psi(a_{k-1}))=(-1)^{n-1}
|\mathcal{M}|. \label{even:5}
\end{split}\end{equation}
Notice that the $(l,i)$ element of matrix $\mathcal{M}$ is
$\Psi_{l}(c_{m+1})-\Psi_{l}(c_{m})$ when  $i=2m$; and
$\Psi_{l}(\tilde{c}_{m+1})-\Psi_{i}(\tilde{c}_{m})$ when $i=2m+1$.
We can rewrite $\tilde{C}_1$ as
$\tilde{C}_1=\tilde{C}_{11}+\tilde{C}_{12}$, where
$\tilde{C}_{11}=(\Psi_{1}(\tilde{c}_{1})-\Psi_{1}(c_1),\ldots,\Psi_{k-1}(\tilde{c}_{1})-\Psi_{k-1}(c_1))'$
and
$\tilde{C}_{12}=(\Psi_{1}(c_{1})-\Psi_{1}(\tilde{c}_0),\ldots,\Psi_{k-1}(c_{1})-\Psi_{k-1}(\tilde{c}_0))'$.
Then $|\mathcal{M}|=|\mathcal{M}_1|+|\mathcal{M}_2|$, where
$\mathcal{M}_1=(\tilde{C}_{11}|\tilde{C}_2|\ldots|\tilde{C}_{k-1})$
and
$\mathcal{M}_2=(\tilde{C}_{12}|\tilde{C}_2|\ldots|\tilde{C}_{k-1}).$
For matrix $\mathcal{M}_1$, by subtracting the 1st column from the
2nd column, subtracting the new 2nd column from the 3rd column,
subtracting the new 3rd column from the 4th column, so on and so
forth until the last column, we have $|\mathcal{M}_1|=D_1(\Psi,
S_{k}(\overline{\tilde{c}_0}))$. As for $\mathcal{M}_2$, we can
rewrite its 3rd column $\tilde{C}_3$ as
$\tilde{C}_3=\tilde{C}_{31}+\tilde{C}_{32}$, where
$\tilde{C}_{31}=(\Psi_{1}(\tilde{c}_{2})-\Psi_{1}(c_2),\ldots,\Psi_{k-1}(\tilde{c}_{2})-\Psi_{k-1}(c_2))'$
and
$\tilde{C}_{32}=(\Psi_{1}(c_{2})-\Psi_{1}(\tilde{c}_1),\ldots,\Psi_{k-1}(c_{2})-\Psi_{k-1}(\tilde{c}_1))'$.
So we have $|\mathcal{M}_2|=|\mathcal{M}_3|+|\mathcal{M}_4|$,
where
$\mathcal{M}_3=(\tilde{C}_{12}|\tilde{C}_2|\tilde{C}_{31}|\tilde{C}_{4}|\ldots|\tilde{C}_{k-1})$
and
$\mathcal{M}_4=(\tilde{C}_{12}|\tilde{C}_2|\tilde{C}_{32}|\tilde{C}_{4}|\ldots|\tilde{C}_{k-1}).$
By similar argument as used for $\mathcal{M}_1$, we have
$|\mathcal{M}_3|=D_1(\Psi, S_{k}(\overline{\tilde{c}_1}))$. So we
have $|\mathcal{M}|=D_1(\Psi,
S_{k}(\overline{\tilde{c}_0}))+D_1(\Psi,
S_{k}(\overline{\tilde{c}_1}))+|\mathcal{M}_4|$. As for
$\mathcal{M}_4$, by subtracting the 3nd column from the 2nd
column, we can have a matrix which has similar patterns as
$\mathcal{M}_2$. Repeating these arguments for every odd column of
each newly generated matrix until we reach the last column, we
will have
\begin{equation}\begin{split}
|\mathcal{M}|=\sum_{j=0}^{n}D_1(\Psi, S_{k}(\overline{\tilde{c}_j})). \label{even:6}
\end{split}\end{equation}
By (\ref{even:q}),  (\ref{even:4}), (\ref{even:5}), and(\ref{even:6}), we have
\begin{equation}\begin{split}
\tilde{\omega}_j=r\frac{D_1(\Psi, S_{k}(\overline{\tilde{c}_j}))}{\sum_{i=0}^{n}D_1(\Psi, S_{k}(\overline{\tilde{c}_i}))}, j=0,\dots,n-1.\label{even:q2}
 \end{split}\end{equation}
Since $D_1(\Psi, S_{k}(\overline{\tilde{c}_i}))>0, i=0,\ldots,n$,
it is clear that $\tilde{\omega}_j,j=0,\ldots,n$ have the same
sign as $r$.

On the other hand, for each $\omega_i$, $i=1,\ldots,n-1,$
\begin{equation}\begin{split}
\omega_i=r\frac{D(\Psi(\tilde{b}_1)-\Psi(\tilde{a}_1),\ldots,
\Psi(\tilde{b}_{k-1})-\Psi(\tilde{a}_{k-1}))}{D(\Psi(b_1)-\Psi(a_1),\ldots,\Psi(b_{k-1})-\Psi(a_{k-1}))}, \label{even:p}
 \end{split}\end{equation}
where $(\tilde{b}_j,\tilde{a}_j)=(b_j,a_j)$ for all $j\neq 2i$ and $(\tilde{b}_{2j},\tilde{a}_{2j})=(\tilde{c}_n,c_n)$. By the similar argument as those for $\tilde{\omega}_j,j=0,\ldots,n$, we can show that
\begin{equation}\begin{split}
\omega_i=r\frac{D_1(\Psi, S_{k}(\overline{c_i}))}{\sum_{j=1}^{n}D_1(\Psi, S_{k}(\overline{c_j}))}, i=1,\dots,n-1.\label{even:p2}
 \end{split}\end{equation}
Since $D_1(\Psi, S_{k}(\overline{c_j}))>0, j=1,\ldots,n$ (Corollary 1), it is clear that $\omega_i,i=1,\ldots,n$ have the same sign as that of $r$. Thus we can draw the conclusion that if at least one number in $\omega_1,\ldots,\omega_n$ and $\tilde{\omega}_0,\ldots,\tilde{\omega}_n$ is positive, then all of them are positive. Next, we will show that (\ref{even:2}) holds.
Notice that (\ref{even:2}) is equivalent to
\begin{equation}\begin{split}
r(\Psi_k(\tilde{c}_n)-\Psi_k(c_n))- (\Psi_k(b_1)-\Psi_k(a_1),\ldots, (\Psi_k(b_{k-1})-\Psi_k(a_{k-1})) \Gamma >0. \label{even:9}
\end{split}\end{equation}
Here, $(b_i, a_i), i=1,\ldots,k-1$ are defined as right before
(\ref{even:3}) and $\Gamma$ is defined as in (\ref{even:3}). By
Theorem 13.3.8 of Harville (1997), the left hand side of
(\ref{even:9}) can be written as
\begin{equation}\begin{split}
r\frac{D(\Psi(b_1)-\Psi(a_1),\ldots,\Psi(b_{k-1})-\Psi(a_{k-1}), \Psi(\tilde{c}_n)-\Psi(c_n))}{D(\Psi(b_1)-\Psi(a_1),\ldots,\Psi(b_{k-1})-\Psi(a_{k-1}))}.
\label{even:10}
\end{split}\end{equation}
Next, we study $D(\Psi(b_1)-\Psi(a_1),\ldots,\Psi(b_{k-1})-\Psi(a_{k-1}), \Psi(\tilde{c}_n)-\Psi(c_n))$. Let $C_i, i=1,\ldots,k$ be the $ith$ column of the corresponding matrix. Let $\tilde{C}_i=(-1)^{i-1}(C_i+C_{i+1}-C_{k})$, $i=1,\ldots,k-2$ and $\tilde{C}_{k-1}=C_{k-1}-C_{k}$. Consider the new matrix
$(\tilde{C}_1|\ldots|\tilde{C}_{k-1}|C_k)$. Then we have
\begin{equation}\begin{split}
D(\Psi(b_1)-\Psi(a_1),\ldots,\Psi(b_{k-1})-\Psi(a_{k-1}), \Psi(\tilde{c}_n)-\Psi(c_n))=(-1)^{n-1}D_1(\Psi,S_{k+1}). \label{even:11}
\end{split}\end{equation}
Here, $S_{k+1}=\{c_1,\ldots,c_n, \tilde{c}_0,\ldots,\tilde{c}_n\}$. Notice that $D_1(\Psi,S_{k+1})>0$  (Corollary 1). Thus,
by (\ref{even:5}), (\ref{even:6}), and (\ref{even:11}), it is clear that (\ref{even:10}) is positive when $r$ is positive. The conclusion follows.
\end{proof}

Proposition \ref{odd} is analogue to Proposition \ref{even},
except that it is for the case that $k$ is odd. It can be proven
using a similar strategy as used in derivation of Proposition
\ref{even}. We omit the proof due to space limit.

\begin{Proposition}\label{odd}
Let $\Psi_1, \ldots,
\Psi_k$ be $k=2n-1$ functions defined on $[A,B]$. Assume that $f_{l,l}(c)>0$, $c\in[A,B]$, $l=1,\ldots,k$.
Suppose for some $c_1,\ldots,c_n$ and $\tilde{c}_1,\ldots,\tilde{c}_{n}$, where
$A\leq c_1<\tilde{c}_1<c_2<\tilde{c}_2<\ldots<c_n<\tilde{c}_n\leq B$, there exist
$\omega_1,\ldots,\omega_n$ and $\tilde{\omega}_1,\ldots,\tilde{\omega}_n$, such that $\sum_{i=1}^n
\omega_i=\sum_{j=1}^n \tilde{\omega}_j$ and the following $k-1$ equations hold:
\begin{equation}\begin{split}
\sum_{i=1}^{n}\omega_i\Psi_l(c_{i}) =\sum_{j=1}^{n}\tilde{\omega}_j\Psi_l(\tilde{c}_{j}),
l=1,\ldots,k-1. \label{odd:1}
\end{split}\end{equation}
If at least one number in $\omega_1,\ldots,\omega_n$ and $\tilde{\omega}_1,\ldots,\tilde{\omega}_n$ is
positive, then all of them should be positive. Under this
situation,
\begin{equation}\begin{split}
\sum_{i=1}^{n}\omega_i\Psi_k(c_{i}) <\sum_{j=1}^{n}\tilde{\omega}_j\Psi_k(\tilde{c}_{j}).
\label{odd:2}
\end{split}\end{equation}
\end{Proposition}

\begin{Corollary}\label{coro1}
Let $\Psi_1, \ldots,
\Psi_k$ be $k$ functions defined on $[A,B]$. Assume that $f_{l,l}(c)>0$, $c\in[A,B]$, $l=1,\ldots,k$.
 Let $z_1<z_2<\ldots<z_t \in [A,B]$ and their associated $r_1,\ldots,r_t$ with $\sum_{i=1}^t r_i=0$ satisfy the following $k-1$ equations:
\begin{equation}\begin{split}
\sum_{i=1}^{t}r_i\Psi_l(z_{i}) =0,
l=1,\ldots,k-1. \label{coro1:1}
\end{split}\end{equation}
Then we have
\begin{itemize}
  \item [(a)] If $t\leq k$, then $r_i=0, i=1,\ldots,t$;
  \item [(b)]  If $t=k+1$ and there exists at least one nonzero $r_i$, then either (i) $r_i>0$, $i$ is odd and $r_i<0$, $i$ is even; or (ii)  $r_i<0$, $i$ is odd and $r_i>0$, $i$ is even.
  \end{itemize}
\end{Corollary}

\begin{proof}
We can rewrite (\ref{coro1:1}) as
\begin{equation}\begin{split}
\sum_{\textit{odd } i}r_i\Psi_l(z_{i}) =\sum_{\textit{even } i}(-r_i)\Psi_l(z_{i}),
l=1,\ldots,k-1. \label{coro1:2}
\end{split}\end{equation}
Notice that  $\sum_{\textit{odd } i}r_i=\sum_{\textit{even } i}(-r_i)$.

When $t\leq k$, consider the first $t-2$ equations of (\ref{coro1:2}), i.e., $l=1,\ldots,t-2$. Suppose there exists at least one nonzero $r_i$. Applying Proposition \ref{even} when $t$ is odd and Proposition \ref{odd} when $t$ is even with $k=t-1$, we have
\begin{equation}\begin{split}
\sum_{\textit{odd } i}r_i\Psi_{t-1}(z_{i})\neq \sum_{\textit{even } i}(-r_i)\Psi_{t-1}(z_{i}).
\end{split}\end{equation}
This contradicts (\ref{coro1:2}). So conclusion (a)
follows.

When $t= k+1$ and there exists at least one nonzero $r_i$, applying Proposition \ref{even} when $t$ is odd and Proposition \ref{odd} when $t$ is even with $k=t-1$, we can draw conclusion (b).
\end{proof}

\subsection{Proof of Lemma 1}

We first study some basic properties of  $\tilde{c}_j$'s assuming
Lemma  \ref{exist} holds.

\begin{Proposition}\label{property}
Assume that Lemma \ref{exist} holds for $k\leq K$. Let $(\tilde{c}_j,
\tilde{\omega}_j)$'s  be the solution set for given
$(c_i,\omega_i)$'s, and $\tilde{c}_0$ or $\tilde{c}_n$ (if
applicable). Let $\omega_i^m$ be a sequence of bounded positive
number for each $i$ and  $(\tilde{c}_j^m, \tilde{\omega}_j^m)$'s
be the solution set with $\omega_i$'s being replaced by
$\omega_i^m$'s and all other values are fixed including
$\tilde{c}_0$ or $\tilde{c}_n$ (if applicable). Then we have
\begin{itemize}
  \item [(i)] $\tilde{c}_j$'s are unique.
  \item [(ii)] If one of the $\omega_i$ sequences, say $\omega_{i_1}$ is an increasing sequence, and all other given values are the same (including $\tilde{c}_0$ or $\tilde{c}_n$ if applicable),
  then the $\tilde{c}_j, j<i_1$ are increasing sequences and   $\tilde{c}_j, j\geq i_1$
  are decreasing sequences.
  On the other hand, if $\omega_{i_1}$ is an decreasing sequence and all other given values are the same,
  then $\tilde{c}_j, j<i_1$ are decreasing sequences and   $\tilde{c}_j, j\geq i_1$
  are increasing sequences.
  \item [(iii)] If $\omega_{i}^m\rightarrow \omega_i$ for all $i$'s, then $\tilde{c}_j^m\rightarrow \tilde{c}_j$ for all
  $j$'s.
  \item [(iv)] If $\omega_{i_1}^m\rightarrow 0$ and $\omega_{i_2}^m \nrightarrow 0$, then either $\underline{\lim}|\tilde{c}_{i_2-1}^m-c_{i_2}|=0$ or
  $\underline{\lim}|\tilde{c}_{i_2}^m-c_{i_2}|=0$.
  \item [(v)] Suppose that $\omega_{i_1}^m\rightarrow 0$. If there exists $i_2>i_1$,
  such that $\underline{\lim}\omega_{i}^m>0$ for $i\geq i_2$,
  then $\tilde{c}_j^m \rightarrow c_{j+1}$ for all $j\geq i_2-1$.
  If there exist $i_3<i_1$, such that $\underline{\lim}\omega_{i}^m>0$ for $i\leq i_3$,
  then $\tilde{c}_j^m \rightarrow c_{j}$ for all
  $j\leq i_3$.
  \item[(vi)] If $\underline{\lim}|\tilde{c}_{j_1}^m-c_{j_1+1}|=0$, then there exists a subsequence $\{m_1\}$ and $i_1(\leq j_1)$, such that
$\lim \omega_{i_1}^{m_1}=0$ and
$\lim|\tilde{c}_{j}^{m_1}-c_{j+1}|=0$ for $i_1\leq j\leq j_1$.
Similarly,  if $\underline{\lim}|\tilde{c}_{j_2}^m-c_{j_2}|=0$,
then there exists a subsequence $\{m_2\}$ and $i_2(> j_2)$, such
that $\lim \omega_{i_2}^{m_2}=0$ and
$\lim|\tilde{c}_{j}^{m_2}-c_{j}|=0$ for $j_2\leq j< i_2$.

\item[(vii)] Suppose that $\omega_i^m<\omega_i$ when $i\leq i_1$
and $\omega_i^m=\omega_i$ otherwise. If $\tilde{c}_{j_1}^m
\rightarrow \tilde{c}_{j_1}$,  for some $j_1\geq i_1$, then
$\omega_i^m\rightarrow \omega_i$ for all $i$'s  and $\tilde{c}_j^m
\rightarrow \tilde{c}_j$ for all
  $j$'s.

\item[(viii)] Let $\tilde{c}_0^m$ be a sequence  numbers between
$\tilde{c}_0$ and $c_1$,  and suppose that $\omega_i^m<\omega_i$
when $i\leq i_1$ and $\omega_{i}^m=\omega_{i}$ otherwise. Let
$(\tilde{c}_j^m, \tilde{\omega}_j^m)$ be the solution set in Case
(a) for given $(c_i,\omega_i^m), i=1,\ldots,n$ and
$\tilde{c}_0^m$. If $\tilde{c}_{n-1}^m\rightarrow
\tilde{c}_{n-1}$, then we must have $\tilde{c}_0^m\rightarrow
\tilde{c}_0$ and $\omega_i^m\rightarrow \omega_i$, $i=1,\ldots,n$.
\end{itemize}
\end{Proposition}

\begin{proof}
(i): Suppose $(\widehat{c}_j,\widehat{\omega}_j)$'s are another solution set. By Lemma \ref{exist}, we have $\sum_j\widehat{\omega}_j+\sum_j(-\tilde{\omega}_j)=0$ and
\begin{equation}\begin{split}
\sum_{j}\widehat{\omega}_j\Psi_l(\widehat{c}_{j})
+\sum_{j}(-\tilde{\omega}_j)\Psi_l(\tilde{c}_{j})=0, l=1,\ldots,k-1.  \label{property:1}
\end{split}\end{equation}
By case by case discussion, notice that the given $\tilde{c}_0$
and $\tilde{c}_n$ (if applicable) are still the same, there are at
most $k$ distinct values among $\widehat{c}_j$'s and
$\tilde{c}_{j}$'s. By rewriting the distinct values as
$z_1<z_2<\ldots<z_t$ ($t\leq k$) and merging the associated
weights into $r_i$, we have $\sum_{i=1}^tr_i=0$ and
\begin{equation}\begin{split}
\sum_{i=1}^{t}r_i\Psi_l(z_{i}) =0,
l=1,\ldots,k-1.
\end{split}\end{equation}
By (a) of Corollary \ref{coro1}, we have  $r_i=0, i=1,\ldots,t$. This implies $\widehat{\omega}_j=\tilde{\omega}_j$ and $\widehat{c}_j=\tilde{c}_j$ for all $j$.

(ii): Suppose $\omega_{i_1}$ increases to $\omega_{i_1}+\delta$,
$\delta>0$, and all other given values are fixed. Let
$\widehat{c}_j$'s and $\widehat{\omega}_j$ be the corresponding
solution set. By Lemma \ref{exist}, we have $\delta+\sum_j
\tilde{\omega}_j+ \sum_j(-\widehat{\omega}_j)=0$ and
\begin{equation}\begin{split}
\delta \Psi_l(c_{i_1})
+\sum_{j}\tilde{\omega}_j\Psi_l(\tilde{c}_{j})+\sum_{j}(-\widehat{\omega}_j)\Psi_l(\widehat{c}_{j})=0, l=1,\ldots,k-1.  \label{property:2}
\end{split}\end{equation}
There are at most $k+1$ distinct values among $c_{i_1}$, $\tilde{c}_{j}$'s and $\widehat{c}_j$'s. Notice that $\delta>0$ and $c_{i_1}$ is distinct from $\tilde{c}_{j}$'s and $\widehat{c}_j$'s,
by (a) of Corollary \ref{coro1}, there are exactly $k+1$ distinct values among $c_{i_1}$, $\tilde{c}_{j}$'s and $\widehat{c}_j$', i.e., they are all distinct except the given $\tilde{c}_0$ and $\tilde{c}_n$ (if applicable). Notice that  $\tilde{\omega}_j$'s and $\widehat{\omega}_j$'s are positive. By (b) of Corollary \ref{coro1}, if we re-order the $k+1$ distinct values,
$\tilde{c}_{j}$'s and $c_{i_1}$ are alternated by $\widehat{c}_{j}$'s. Thus we have $\tilde{c}_j<\widehat{c}_j, j<i_1$ and $\tilde{c}_j>\widehat{c}_j, j\geq i_1$. When $\omega_{i_1}$ decreases, the proof is a completely analogous.

(iii):  Since that $\tilde{c}_j^m$'s are bounded with $c_{j}$ and
$c_{j+1}$, so are $\overline{\lim}\tilde{c}_j^m$ and
$\underline{\lim} \tilde{c}_j^m$.  It is sufficient to show that
$\overline{\lim} \tilde{c}_j^m=\underline{\lim}
\tilde{c}_j^m=\tilde{c}_j$ for all $j$'s. Here, we show that the
conclusion holds for one $j$, say $j_1$. The proof for other cases
are complete analogy. There exists a subsequence of $\{m\}$, say,
$\{m_1\}$, such that  $\lim \tilde{c}_{j_1}^{m_1}= \overline{\lim}
\tilde{c}_{j_1}^m$. Then we can further choose a subsequence from
$\{m_1\}$, say,  $\{m_2\}$, such that  $\lim \tilde{c}_{j}^{m_2}$
exists for one $j\ne j_1$. We can continue this way to choose a
subsequence until we have a subsequence, say, $\{m'\}$, such that
$\lim \tilde{c}_{j_1}^{m'}= \overline{\lim} \tilde{c}_{j_1}^m$,
$\lim \tilde{c}_{j}^{m'}$ exists for all $j\ne j_1$, and $\lim
\tilde{\omega}_{j}^{m'}$ exists for all $j$. Then by Lemma
\ref{exist} and $\lim \omega_{i}^{m'}\rightarrow \omega_{i}$, we
have
\begin{equation}\begin{split}
&\sum_{i} \omega_i=\lim \tilde{\omega}_{j_1}^{m'}+\sum_{j\ne j_1}\lim \tilde{\omega}_j^{m'};\\
&\sum_{i}\omega_i\Psi_l(c_{i})=\lim \tilde{\omega}_{j_1}^{m'}\Psi_l( \overline{\lim}\tilde{c}_{j_1}^{m} )+\sum_{j\ne j_1}\lim \tilde{\omega}_j^{m'}\Psi_l(\lim \tilde{c}_{j}^{m'}),
l=1,\ldots,k-1. \label{property:iii}\\
\end{split}\end{equation}
By the uniqueness of $\tilde{c}_j$'s from (i) and (\ref{property:iii}), we must have $\overline{\lim}\tilde{c}_{j_1}^{m}=\tilde{c}_{j_1}$. Similarly we can show that $\underline{\lim}\tilde{c}_{j_1}^{m}=\tilde{c}_{j_1}$. The conclusion follows.

(iv): By similar arguments as used in (iii), we can find a subsequence  $\{m_1\}$, such that
 $\lim \omega_{i_1}^{m_1}=0$,  $\lim \omega_{i_2}^{m_1}>0$, $\lim \omega_{i}^{m_1}$ exists
 for all $i\ne i_1$, or  $i_2$, and $\lim \tilde{c}_j^{m_1}$ and $\lim \tilde{\omega}_j^{m_1}$ exist for all $j$'s. Then by Lemma \ref{exist}, we have
\begin{equation}\begin{split}
&\sum_{i\ne i_1} \lim \omega_i^{m_1}+\sum_j(-\lim \tilde{\omega}_j^{m_1})=0, \\
&\sum_{i\ne i_1}\lim \omega_i^{m_1} \Psi_l(c_i)
+\sum_{j}(-\lim \tilde{\omega}_j^{m_1})\Psi_l(\lim \tilde{c}_{j}^{m_1})=0, l=1,\ldots,k-1.  \label{property:iv}
\end{split}\end{equation}
There are at most $k$ distinct values among $c_{i}$'s, $i\ne i_1$, and $\lim \tilde{c}_{j}^{m_1}$'s. By  similar arguments as used in the proof of (i), we have either $\lim \tilde{c}_{i_2-1}^{m_1}=c_{i_2}$ or $\lim \tilde{c}_{i_2}^{m_1}=c_{i_2}$. The conclusion follows.

(v): Suppose $\underline{\lim}\omega_{i}^m>0$ for $i\geq i_2$. If $\tilde{c}_j^m \nrightarrow c_{j+1}$ for some $j\geq i_2-1$, say, $j_1$, then there exists a subsequence $\{m_1\}$, such that
$\lim\omega_{i}^{m_1}$ exists for all $i$'s, $\lim \tilde{c}_j^{m_1}$ and $\lim \tilde{\omega}_j^{m_1}$ exist for all $j$'s, $\lim\omega_{i_1}^{m_1}=0$, $\lim\omega_{i}^{m_1}>0$ for $i\geq i_2-1$,  and  $\lim \tilde{c}_{j_1}^{m_1}\ne c_{j_1+1}$. By similar arguments as used in the proof of (i), we have $\lim \tilde{c}_{j}^{m_1}=c_{j+1}$ for $j\geq i_2-1$ including $j_1$. This is a contradiction.  The case when $\underline{\lim}\omega_{i}^m>0$ for $i\leq  i_3$ can be proven analogously.

(vi): Suppose $\underline{\lim}|\tilde{c}_{j_1}^m-c_{j_1+1}|=0$.
By similar argument as used in (iii), we can find a subsequence
$\{m_1\}$, such that $\lim \tilde{c}_{j_1}^{m_1}=c_{j_1+1}$, $\lim
\tilde{c}_j^{m_1}$ and $\lim \tilde{\omega}_j^{m_1}$ exist for all
$j$,  $\lim \omega_{i}^{m_1}$ exists for all $i$. Then,
$\lim\omega_i^{m_1}=0$ for some $i\leq j_1$. Otherwise, suppose
that $\lim\omega_i^{m_1}>0$ for $i\leq j_1$. If there is one
$i>j_1$ such that $\lim \omega_{i}^{m_1}=0$, then by similar
arguments as used in the proof of (i), we must have $\lim
\tilde{c}_{j_1}^{m_1}=c_{j_1}$. This is a contradiction. Therefore
we must have $\lim \omega_{i}^{m_1}>0$ for all $i$'s, so $\lim
\tilde{c}_{j_1}^{m_1}=c_{j_1}^*$. Here, $(\tilde{c}_j^*,
\tilde{\omega}_j^*)$ is the corresponding  solution for given
$(c_i,\lim \omega_i^{m_1})$. So we have $c_{j_1}^*<c_{j_1+1}$.
This is also a contradiction. Thus we must have
$\lim\omega_i^{m_1}=0$ for some $i\leq j_1$. Let $i_1$ be the
largest $i$. If $i_1=j_1$, obviously the conclusion holds. If
$i_1<j_1$, by the definition of $i_1$, we have
$\lim\omega_i^{m_1}>0$ when $i_1<i\leq j_1$. By (iv) and the fact
that $\lim \tilde{c}_j^{m_1}$ exists for all $j$'s, we have either
$\lim\tilde{c}_{i-1}^{m_1}=c_i$ or  $\lim\tilde{c}_{i}^{m_1}=c_i$
for $i_1<i\leq j_1$. But we have $\lim
\tilde{c}_{j_1}^{m_1}=c_{j_1+1}$, which implies that $\lim
\tilde{c}_{i}^{m_1}=c_{i+1}$ for $i_1\leq i\leq j_1$. The
conclusion follows. The case when
$\underline{\lim}|\tilde{c}_{j_2}^m-c_{j_2}|=0$ is completely
analogous.

(vii):  Suppose that $\omega_i^m\nrightarrow \omega_i$ for some
$i$, say $i_2$, then there must exist a subsequence $\{m_1\}$,
such that $\lim\omega_i^{m_1}$, $\lim\tilde{\omega}_j^{m_1}$, and
$\lim\tilde{c}_j^{m_1}$ all exist, and
$\lim\omega_{i_2}^{m_1}<\omega_{i_2}$. We have
$\lim\tilde{c}_j^{m_1}=\tilde{c}_{j}'$, where
$(\tilde{c}_{j}',\tilde{\omega}_j)$'s are the solution set for
given $(c_i,\lim\omega_i^{m_1})$, $i=1,\ldots,i_1$ and
$(c_i,\omega_i)$, $i>i_1$. However,  $\lim\omega_{i}^{m_1}\leq
\omega_{i}$ for $i\leq i_1$ and
$\lim\omega_{i_2}^{m_1}<\omega_{i_2}$, we must have
$\tilde{c}_{j_1}'>\tilde{c}_{j_1}$ by (ii). This is contradictory
to $\tilde{c}_{j_1}^m \rightarrow \tilde{c}_{j_1}$.

(viii): Notice that in Case (a), if $(c_i, \omega_i),
i=1,\ldots,n$ hold the same value, we can show that (similar to
the proof of (ii)) the solution $\tilde{c}_n$ is an increasing
function of $\tilde{c}_0$ as long as $\tilde{c}_0<c_1$. Apply this
and use a similar approach as in (vii), the conclusion follows.
\end{proof}

For the ease of presentation, we define
$(C,\Omega)=\{(c_i,\omega_i)\}$, where $c_i<c_{i+1}$ and
$\omega_i>0$;
$(\widetilde{C},\widetilde{\Omega})=\{(\tilde{c}_j,\tilde{\omega}_j)\}$,
where $\tilde{c}_j<\tilde{c}_{j+1}$ and $\tilde{\omega}_j>0$;
$G_l(C,\Omega,\widetilde{C},\widetilde{\Omega})=\sum_{i}\omega_i\Psi_l(c_{i})-\sum_{j}\tilde{\omega}_j\Psi_l(\tilde{c}_{j})$.
For given $(C,\Omega)$ with appropriate cardinality, let
\begin{itemize}
  \item [(i)] $S^{I}_j(C,\Omega,\tilde{c}_0)=\tilde{c}_j$, where $\tilde{c}_j$'s  are given under Case (a);
  \item [(ii)] $S^{II}_j(C,\Omega,\tilde{c}_n)=\tilde{c}_j$, where $\tilde{c}_j$'s are  given under Case (b);
    \item [(iii)] $S^{III}_j(C,\Omega,\tilde{c}_0,\tilde{c}_n)=\tilde{c}_j$, where $\tilde{c}_j$'s are given under Case (c);
  \item [(iv)] $S^{IV}_j(C,\Omega)=\tilde{c}_j$, where $\tilde{c}_j$'s are  given under Case (d).
\end{itemize}

We will use mathematical induction to prove Lemma \ref{exist}.
When $k=2$ and 3, Lemma \ref{exist} has been proven by Yang and
Stufken (2009). We use the following two propositions to prove
Lemma \ref{exist} for arbitrary $k$, i.e., (i) assume Lemma
\ref{exist} holds when $k\leq 2n-2$, then show that it also holds
for $k=2n-1$ and (ii) assume Lemma \ref{exist} holds when $k\leq
2n-1$, then show that it also holds for $k=2n$.  Notice that once
we can show there exist such $ (\widetilde{C},
\widetilde{\Omega})$ which satisfies (\ref{exist:1}) and
(\ref{exist:2}), then by either Proposition \ref{even} or
\ref{odd}, the inequality in (\ref{exist:3}) holds. To prove the
existence of $ (\widetilde{C}, \widetilde{\Omega})$,  the strategy
is similar for each of the four cases.

\begin{Proposition}\label{dedu2}
If Lemma \ref{exist} holds for $k\leq 2n-2$, then it will also
hold for $k=2n-1$.
\end{Proposition}

\begin{proof}
If Case (a) holds, we can consider a new function set
$\widetilde{\Psi}_1, \ldots, \widetilde{\Psi}_k$ on $[-B,-A]$.
Here $\widetilde{\Psi}_i(c)=-\Psi_i(-c)$ when $i$ is odd and
$\widetilde{\Psi}_i(c)=\Psi_i(-c)$ when $i$ is even. For the new
function set,  we can verify that  the corresponding $f_{l,l}>0$,
$c\in [-B, -A]$, $l=1,\ldots,k$. Let $C^{-}=\{-c_i,
i=1,\ldots,n\}$ and $\tilde{c}_0^{-}=-\tilde{c}_n$. Apply Case (a)
to the new function set  $\widetilde{\Psi}_1, \ldots,
\widetilde{\Psi}_k$ with $C^{-}$ and $\tilde{c}_0^{-}$, we obtain
the solution set $\widetilde{C}^{-}=\{\tilde{c}_j^-,
j=0,\ldots,n-1\}$. Let
$\widetilde{C}=\{\tilde{c}_j=-\tilde{c}_{n-j}^-, j=1,\ldots,n\}$,
then  Case (b) follows by replacing  $\widetilde{\Psi}_i$ with
$\Psi_i$. So we only need to prove Case (a).

In this case, $(C,\Omega)=\{(c_i,\omega_i), i=1,\ldots,n\}$ and
$\tilde{c}_0$ are given.  It is sufficient to show that there
exists a solution set
$(\widetilde{C},\widetilde{\Omega})=\{(\tilde{c}_j,\tilde{\omega}_j),
j=0,\ldots,n-1\}$ that satisfy (\ref{exist:1}) and (\ref{exist:2})
of Lemma \ref{exist}.

Let $D=\{d_1, \ldots, d_{n-1}, \omega_n\}$, where $d_i\in
(0,\omega_i)$, $i=1,\ldots,n-1$. Define
$\Omega^{\overline{D}}=\{\omega_1-d_1,\ldots,\omega_{n-1}-d_{n-1}\}$
and $C^{-n}=\{c_1,\ldots,c_{n-1}\}$. We are going to show that for
any given $d_{n-1}\in$ $(0,\omega_{n-1})$, there exist $d_i$,
$i=1,\ldots,n-2$, and $\tilde{c}_{n-1}\in (c_{n-1},c_n)$, such
that
\begin{equation}\begin{split}
S^{III}_{j}(C^{-n},\Omega^{\overline{D}},\tilde{c}_0,\tilde{c}_{n-1})=S^{IV}_j(C,D), \label{dedu2:0}
\end{split}\end{equation}
for $j=1,\ldots,n-1$. Once we show that (\ref{dedu2:0}) holds, we
can let $\tilde{\omega}'_j$ be the corresponding weight of
$S^{III}_{j}(C^{-n},\Omega^{\overline{D}},\tilde{c}_0,\tilde{c}_{n-1})$,
$j=0,\ldots,n-1$ and $\tilde{d}_j$ be the corresponding weight of
$S^{IV}_j(C,D)$, $j=1,\ldots,n-1$. Define
$\widetilde{C}=\{\tilde{c}_0, S^{IV}_j(C,D), j=1,\ldots, n-1\}$
and  $\widetilde{\Omega}=\{\tilde{\omega}'_0,
\tilde{\omega}'_j+\tilde{d}_j,j=1,\ldots,n-1\}$, then
\begin{equation}\begin{split}
G_l(C,\Omega,\widetilde{C},\widetilde{\Omega})=0, l=1,\ldots,2n-3. \label{dedu2:1}
\end{split}\end{equation}
It can be shown that $G_{2n-2}(C,\Omega,\widetilde{C},\widetilde{\Omega})$ is a continuous function of $d_{n-1}$. If we further show that $G_{2n-2}(C,\Omega,\widetilde{C},\widetilde{\Omega})$ has different signs when $d_{n-1}\downarrow 0$ and $d_{n-1}\uparrow \omega_{n-1}$, then there must exists a $d_{n-1}\in
(0,\omega_{n-1})$, such that $G_{2n-2}(C,\Omega,\widetilde{C},\widetilde{\Omega})=0$.
Then our conclusion follows.

This strategy will be achieved in three steps: (i) for any given
$d_{n-1}\in (0,\omega_{n-1})$ and
$\tilde{c}_{n-1}\in(c_{n-1},c_n)$, there exists $d_i$,
$i=1,\ldots,n-2$, such that (\ref{dedu2:0}) holds for
$j=1,\ldots,n-2$; (ii) for any given $d_{n-1}\in
(0,\omega_{n-1})$, there exist $\tilde{c}_{n-1}\in (c_{n-1},c_n)$
and $d_i$, $i=1,\ldots,n-2$, such that (\ref{dedu2:0}) holds for
$j=1,\ldots,n-1$; (iii)
$G_{2n-2}(C,D,\widetilde{C},\widetilde{\Omega})$ has different
signs when $d_{n-1}\downarrow 0$ and $d_{n-1}\uparrow
\omega_{n-1}$.

 Step (i) can be proven by mathematical induction. We
first show that for any given $d_i\in (0,\omega_i),
i=2,\ldots,n-1$ and $\tilde{c}_{n-1}\in (c_{n-1},c_n)$, there
exists $d_1\in (0,\omega_1)$, such that (\ref{dedu2:0}) holds when
$j=1$. This is because when $d_1\uparrow \omega_1$, we have
\begin{equation}\begin{split}
S^{III}_1(C^{-n},\Omega^{\overline{D}},\tilde{c}_0,\tilde{c}_{n-1}) \rightarrow c_2 \text{ and }S^{IV}_1(C,D) \rightarrow S^{IV}_1(C,D') <c_2, \label{dedu2:2}
\end{split}\end{equation}
where $D'=\{\omega_1,d_2, \ldots, d_{n-1}, \omega_n\}$. This is
due to (v) and (iii) of Proposition \ref{property}, respectively.
When $d_1\downarrow 0$, we have
\begin{equation}\begin{split}
S^{III}_1(C^{-n},\Omega^{\overline{D}},\tilde{c}_0,\tilde{c}_{n-1}) \rightarrow
S^{III}_1(C^{-n},\Omega',\tilde{c}_0,\tilde{c}_{n-1})<c_2 \text{ and }S^{IV}_1(C,D) \rightarrow c_2, \label{dedu2:3}
\end{split}\end{equation}
where $\Omega'=\{\omega_1,\omega_2-d_2, \ldots,
\omega_{n-1}-d_{n-1}\}$. This is due to (iii) and (v) of
Proposition \ref{property}, respectively. By (\ref{dedu2:2}) and
(\ref{dedu2:3}), it is clear that\\
$S^{III}_1(C^{-n},\Omega^{\overline{D}},\tilde{c}_0,\tilde{c}_{n-1})
-S^{IV}_1(C,D)$ is positive when $d_1\uparrow \omega_1$ and
negative when $d_1\downarrow 0$. It can be shown that
$S^{III}_1(C^{-n},\Omega^{\overline{D}},\tilde{c}_0,\tilde{c}_{n-1})
-S^{IV}_1(C,D)$ is a continuous function of $d_1$. Thus there
exists a point $d_1$ such that (\ref{dedu2:0}) holds when $j=1$.
Notice that $d_1$ depends on $d_i$'s, $i=2,\ldots,n-1$ and
$\tilde{c}_{n-1}$.

Now, we assume that for any given $d_i\in (0,\omega_i), i=p(\leq
n-2),\ldots,n-1,$ and $\tilde{c}_{n-1}\in (c_{n-1},c_n)$, there
exists $d_i\in (0,\omega_i), i=1,\ldots,p-1$, such that
(\ref{dedu2:0}) holds when $j=1, \ldots,p-1$. Consider any given
$d_i\in (0,\omega_i), i=p+1,\ldots,n-1$ and $\tilde{c}_{n-1}$. By
assumption, for any $d_{p}\in (0,\omega_p)$,  there exists $d_i\in
(0,\omega_i), i=1,\ldots,p-1$ such that (\ref{dedu2:0}) holds when
$j=1, \ldots,p-1$. When $d_p\downarrow 0$, by (v) of Proposition
\ref{property}, we have
\begin{equation}\begin{split}
S^{IV}_p(C,D) \rightarrow c_{p+1}.
\label{dedu2:4}
\end{split}\end{equation}
Next, we are going to show that  $d_i\rightarrow 0$,
$i=1,\ldots,p-1$ when $d_p\downarrow 0$. Suppose there exists some
$i(<p)$, such that $d_i\nrightarrow 0$. Let $i_1$ be the smallest
$i$ that satisfies this condition. If $i_1=1$, then we have
$\underline{\lim}|S^{IV}_1(C,D)-c_1|=0$ by (iv) of Proposition
\ref{property}. This implies that
$\underline{\lim}|S^{III}_{1}(C^{-n},\Omega^{\overline{D}},\tilde{c}_0,\tilde{c}_{n-1})-c_{i_1}|=0$
since (\ref{dedu2:0}) holds for $j=1, \ldots,p-1$. If $i_1>1$,
then by (iv) of Proposition \ref{property}, we have either
$\underline{\lim}|S^{IV}_{i_1-1}(C,D)-c_{i_1}|=0$ or
$\underline{\lim}|S^{IV}_{i_1}(C,D)-c_{i_1}|=0$. Suppose that
$\underline{\lim}|S^{IV}_{i_1-1}(C,D)-c_{i_1}|=0$, then we have
$\underline{\lim}|S^{III}_{i_1-1}(C^{-n},\Omega^{\overline{D}},\tilde{c}_0,\tilde{c}_{n-1})-c_{i_1}|=0$
by (\ref{dedu2:0}) again. By (vi) of Proposition \ref{property},
$\underline{\lim}(\omega_i-d_i)=0$ for some $i\leq i_1-1$. By the
definition of $i_1$, we have $d_i\rightarrow 0$ for $i< i_1$. This
is a contradiction. So we must have
$\underline{\lim}|S^{IV}_{i_1}(C,D)-c_{i_1}|=0$, which implies
that
$\underline{\lim}|S^{III}_{i_1}(C^{-n},\Omega^{\overline{D}},\tilde{c}_0,\tilde{c}_{n-1})
-c_{i_1}|=0$ by (\ref{dedu2:0}).

By (vi) of Proposition \ref{property}, there exists a subsequence
of $\{d_p\downarrow 0\}$ and $i_2>i_1$ such that $\lim
(\omega_{i_2}-d_{i_2})=0$ and
$\lim|S^{III}_{i}(C^{-n},\Omega^{\overline{D}},\tilde{c}_0,\tilde{c}_{n-1})
-c_{i}|=0$ for $i_1\leq i<i_2$. For this subsequence,  by (iv) of
Proposition \ref{property} and the fact that $\lim
d_{i_2}=\omega_{i_2}$, we have either
$\underline{\lim}|S^{IV}_{i_2-1}(C,D)-c_{i_2}|=0$ or
$\underline{\lim}|S^{IV}_{i_2}(C,D)-c_{i_2}|=0$. However, $\lim
S^{IV}_{i_2-1}(C,D)$=$\lim
S^{III}_{i_2-1}(C^{-n},\Omega^{\overline{D}},\tilde{c}_0,\tilde{c}_{n-1})=c_{i_2-1}$.
Thus we must have $\underline{\lim}|S^{IV}_{i_2}(C,D)-c_{i_2}|=0$,
this also implies that\\
$\underline{\lim}|S^{III}_{i_2}(C^{-n},\Omega^{\overline{D}},\tilde{c}_0,\tilde{c}_{n-1})-c_{i_2}|=0$.

By the exact same argument, there must exist $i_3>i_2$ and a
subsequence, such that $\lim (\omega_{i_3}-d_{i_3})=0$ and
$\underline{\lim}|S^{III}_{i_3}(C^{-n},\Omega^{\overline{D}},\tilde{c}_0,\tilde{c}_{n-1})-c_{i_3}|=0$.
Repeat this argument again, we can find strictly increasing
numbers $i_4, i_5,\ldots$ and each of them has the same property
as $i_2$ and $i_3$. Since $p$ is finite, one of
$\{i_2,i_3,i_4,\ldots\}$  must be greater than or equal to $p$.
This leads to a contradiction since all $d_i(<\omega_i)$, $i>p$
are fixed and $d_p\downarrow 0$. Thus we have $d_i\rightarrow 0$,
$i=1,\ldots,p-1$ when $d_p\downarrow 0$. This implies that
\begin{equation}\begin{split}
S^{III}_{p}(C^{-n},\Omega^{\overline{D}},\tilde{c}_0,\tilde{c}_{n-1}) \rightarrow S^{III}_{p}(C^{-n},\Omega',\tilde{c}_0,\tilde{c}_{n-1})<c_{p+1},
\label{dedu2:5}
\end{split}\end{equation}
where $\Omega'=\{\omega_1,\ldots, \omega_{p},\omega_{p+1}-d_{p+1},\ldots, \omega_{n-1}-d_{n-1}\}$. (\ref{dedu2:5}) is due to (iii) of Proposition \ref{property}. By (\ref{dedu2:4}) and (\ref{dedu2:5}), we have
\begin{equation}\begin{split}
S^{IV}_p(C,D)-S^{III}_{p}(C^{-n},\Omega^{\overline{D}},\tilde{c}_0,\tilde{c}_{n-1}) >0
\label{dedu2:6}
\end{split}\end{equation}
when $d_p\downarrow 0$.

On the other hand, notice that $d_p\uparrow \omega_p$ is
equivalent to $\omega_p-d_p\downarrow 0$. We can show that the
inequality sign in (\ref{dedu2:6}) will reverse using an analogous
approach as used in the case of  $d_p\downarrow 0$. Due to space
limit, we will just give the outline of the proof here. First, we
have
\begin{equation}\begin{split}
S^{III}_p(C^{-n},\Omega^{\overline{D}},\tilde{c}_0,\tilde{c}_{n-1}) \rightarrow c_{p+1}.
\label{dedu2:7}
\end{split}\end{equation}
Next, we are going to show that  $d_i\rightarrow \omega_i$,
$i=1,\ldots,p-1$ when $d_p\uparrow \omega_p$. Suppose there exists
some $i(<p)$, such that $d_i\nrightarrow \omega_i$. Let $i_1$ be
the smallest one. We can show that
$\underline{\lim}|S^{IV}_{i_1}(C,D) -c_{i_1}|=0$. This implies
that there exists a subsequence of $\{d_p\uparrow \omega_p\}$ and
$i_2>i_1$ such that $\lim d_{i_2}=0$ and
$\underline{\lim}|S^{IV}_{i_2}(C,D)-c_{i_2}|=0$. Repeat this
argument, we can find strictly increasing numbers $i_3,
i_4,\ldots$ and each of them has the same property as $i_2$. Since
$p$ is finite, one of $\{i_2,i_3,i_4,\ldots\}$  must be greater
than or equal to $p$. This leads to a contradiction since all
$d_i$, $i>p$ are fixed and $d_p\uparrow \omega_p>0$. Thus we have
$d_i\rightarrow \omega_i$, $i=1,\ldots,p-1$ when $d_p\uparrow
\omega_p$. This implies that
\begin{equation}\begin{split}
S^{IV}_p(C,D) \rightarrow S^{IV}_p(C,D')<c_{p+1},
\label{dedu2:8}
\end{split}\end{equation}
where $D'=\{\omega_1,\ldots, \omega_{p},d_{p+1},\ldots, d_{n-1}, \omega_n\}$. By (\ref{dedu2:7}) and (\ref{dedu2:8}), we have
\begin{equation}\begin{split}
S^{IV}_p(C,D)-S^{III}_{p}(C^{-n},\Omega^{\overline{D}},\tilde{c}_0,\tilde{c}_{n-1}) <0
\label{dedu2:9}
\end{split}\end{equation}
when $d_p\uparrow \omega_p$. It can be shown that
$S^{IV}_p(C,D)-S^{III}_{p}(C^{-n},\Omega^{\overline{D}},\tilde{c}_0,\tilde{c}_{n-1})$
is a continuous function of $d_p$. By (\ref{dedu2:6}) and
(\ref{dedu2:9}), there must exists $d_p$, such that
\begin{equation}\begin{split}
S^{IV}_p(C,D)=S^{III}_{p}(C^{-n},\Omega^{\overline{D}},\tilde{c}_0,\tilde{c}_{n-1}).
\label{dedu2:10}
\end{split}\end{equation}
By mathematical induction, we have shown step (i).

Now we are going to prove step (ii). Let
$\Omega_{d_{n-1}}=\{\omega_1,\ldots,\omega_{n-2},d_{n-1},\omega_n\}$.
Applying Lemma \ref{exist} when $k=2n-2$, we have
$\tilde{c}^*(d_{n-1})\in (c_{n-1},c_n)$, where
$\tilde{c}^*(d_{n-1})=S^{IV}_{n-1}(C,\Omega_{d_{n-1}})$. From step
(i), we know that for any given $d_{n-1}$ and $\tilde{c}_{n-1}\in
(\tilde{c}^*(d_{n-1}),c_n)$, there exists $d_i$, $i=1,\ldots,n-2$,
such that (\ref{dedu2:0}) holds for $j=1,\ldots,n-2$. It can be
shown that $S^{IV}_{n-1}(C,D)$ is a continuous function of
$\tilde{c}_{n-1}$. Then it is sufficient to show that
$\underline{\lim} S^{IV}_{n-1}(C,D)>\tilde{c}_{n-1}$ when
$\tilde{c}_{n-1}\downarrow \tilde{c}^*(d_{n-1})$ and
$\overline{\lim}S^{IV}_{n-1}(C,D)<\tilde{c}_{n-1}$ when
$\tilde{c}_{n-1}\uparrow c_n$.

Suppose that $\underline{\lim} S^{IV}_{n-1}(C,D)\leq
\tilde{c}_{n-1}$ when $\tilde{c}_{n-1}\downarrow
\tilde{c}^*(d_{n-1})$. There exists a subsequence of
$\tilde{c}_{n-1}\downarrow \tilde{c}^*(d_{n-1})$, such that $\lim
S^{IV}_{n-1}(C,D)\leq \tilde{c}_{n-1}$. Since $d_i<\omega_i$,
$i=1,\ldots,n-2$, by (ii) of Proposition \ref{property}, we must
have $S^{IV}_{n-1}(C,D)>
S^{IV}_{n-1}(C,\Omega_{d_{n-1}})=\tilde{c}^*(d_{n-1})$. This
implies that for the subsequence of $\tilde{c}_{n-1}\downarrow
\tilde{c}^*(d_{n-1})$, $\lim S^{IV}_{n-1}(C,D)=
\tilde{c}^*(d_{n-1})$. By (vii) of Proposition \ref{property} and
the fact that $d_i<\omega_i$, $i=1,\ldots,n-2$, we have
$d_i\rightarrow \omega_i$, $i=1,\ldots,n-2$. Consequently, we have
\begin{equation}\begin{split}
S^{IV}_{n-2}(C,D)\rightarrow S^{IV}_{n-2}(C,\Omega_{d_{n-1}})<c_{n-1} \text{ and } S^{III}_{n-2}(C^{-n},\Omega^{\overline{D}},\tilde{c}_0,\tilde{c}_{n-1})\rightarrow c_{n-1},
\end{split}\end{equation}
by  (iii) and (v) of Proposition \ref{property}, respectively.
This is a contradiction to (\ref{dedu2:0}) when $j=n-2$.

Suppose that $\overline{\lim}S^{IV}_{n-1}(C,D)\geq
\tilde{c}_{n-1}$ when $\tilde{c}_{n-1}\uparrow c_n$. There exists
a subsequence of $\tilde{c}_{n-1}\uparrow c_n$, such that $\lim
S^{IV}_{n-1}(C,D)\geq \tilde{c}_{n-1}$. By the assumption that
Lemma \ref{exist} holds when $k=2n-2$, we have
$S^{IV}_{n-1}(C,D)<c_n$. This implies that for the subsequence of
$\tilde{c}_{n-1}\uparrow c_n$, $S^{IV}_{n-1}(C,D)\rightarrow c_n$.
By (vi) of Proposition \ref{property}, there exists a
sub-subsequence of $\tilde{c}_{n-1}\uparrow c_n$ and $i_1(<n-1)$
(notice that $d_{n-1}$ is fixed), such that $\lim d_{i_1}=0$ and
$\lim S^{IV}_{j}(C,D)=c_{j+1}$ for $i_1\leq j\leq n-1$. From the
proof of step (i), $\lim d_{i_1}=0$ means $\lim d_{i}=0$ for
$i\leq i_1$. On the other hand, we have
$S^{III}_{i_1}(C^{-n},\Omega^{\overline{D}},\tilde{c}_0,\tilde{c}_{n-1})=c_{i_1+1}$
by (\ref{dedu2:0}) holds for $j=i_1$.  By (vi) of Proposition
\ref{property}, there exists a sub-sub-subsequence of
$\tilde{c}_{n-1}\uparrow c_n$ and $i_2(\leq i_1)$, such that $\lim
d_{i_2}=\omega_{i_2}$. This is a contradiction to $\lim
d_{i_2}=0$. This proves step (ii).

Now we are going to prove step (iii). By similar arguments as used
in the proof of $\lim d_{p}=0$ in step (i), which implies $\lim
d_{i}=0$ for $i\leq p$, we can show that $\lim d_{i}=0$ for
$i=1,\ldots,n-1$ when $d_{n-1}\downarrow 0$. Recall the definition
of $(\widetilde{C},\widetilde{\Omega})$ at the beginning of the
proof,
\begin{equation}\begin{split}
G_{2n-2}(C,\Omega,\widetilde{C},\widetilde{\Omega})\rightarrow G_{2n-2}(C^{-n},\Omega^{-n},\widetilde{C}^{-n},\widetilde{\Omega}^{-n})<0. \label{dedu2:11}
\end{split}\end{equation}
Here,  $(\widetilde{C}^{-n},\widetilde{\Omega}^{-n})$ is the solution set of $(C^{-n},\Omega^{-n})=\{(c_i,\omega_i), i=1,\ldots,n-1\}$ with  given
$\tilde{c}_0$ and $\tilde{c}_n(=c_n)$ under Case (c) of Lemma \ref{exist} when $k=2n-2$.

Similarly, we can show that $\lim d_{i}=\omega_i$ for
$i=1,\ldots,n-1$ when $d_{n-1}\uparrow \omega_{n-1}$. Therefore,
we have
\begin{equation}\begin{split}
G_{2n-2}(C,\Omega,\widetilde{C},\widetilde{\Omega})\rightarrow G_{2n-2}(C,\Omega,\widetilde{C}',\widetilde{\Omega}')>0. \label{dedu2:12}
\end{split}\end{equation}
Here,  $(\widetilde{C}',\widetilde{\Omega}')$ is the solution set
of $(C, \Omega)$ under Case (d) of Lemma \ref{exist} while
$k=2n-2$. (\ref{dedu2:11}) and (\ref{dedu2:12}) give the proof of
step (iii). This completes the proof proposition \ref{dedu2}.

\end{proof}

\begin{Proposition}\label{dedu1}
If Lemma \ref{exist} holds when $k\leq 2n-1$. Then it also holds when $k=2n$.
\end{Proposition}
\begin{proof}
We first prove that Case (c) holds. In this case,\\
$(C,\Omega)=\{(c_i,\omega_i), i=1,\ldots,n\}$, $\tilde{c}_0$ and
$\tilde{c}_n$ are given. It is sufficient to show that there
exists a solution set
$(\widetilde{C},\widetilde{\Omega})=\{(\tilde{c}_j,\tilde{\omega}_j),
j=0,\ldots,n\}$ which satisfies (\ref{exist:1}) and
(\ref{exist:2}) of Lemma \ref{exist}. The proof is similar to that
in  Proposition \ref{dedu2}. Here we will provide an
outline of the proof.

Define $D=\{d_i\in (0,\omega_i),i=1,\ldots,n\}$ and $\Omega^{\overline{D}}=\{\omega_i-d_i, i=1,\ldots,n\}$. We are going to show that for any given $d_n\in(0,\omega_n)$, there exist $d_i\in (0,\omega_i)$, $i=1,\ldots,n-1$, such that
\begin{equation}\begin{split}
S^{I}_j(C,\Omega^{\overline{D}},\tilde{c}_0)=S^{II}_j(C,D,\tilde{c}_n) \label{dedu1:0}
\end{split}\end{equation}
for $j=1,\ldots,n-1$. Once we show that (\ref{dedu1:0}) holds, we
can let $\tilde{\omega}'_j$ be the corresponding weight of
$S^{I}_j(C,\Omega^{\overline{D}},\tilde{c}_0)$, $j=0,\ldots,n-1$
and $\tilde{d}_j$ be the corresponding weight of
$S^{II}_j(C,D,\tilde{c}_n)$, $j=1,\ldots,n$. Then define
$\widetilde{C}=\{\tilde{c}_0,
S^{I}_j(C,\Omega^{\overline{D}},\tilde{c}_0), j=1,\ldots,n-1,
\tilde{c}_n\}$ and $\widetilde{\Omega}=\{\tilde{\omega}'_0,
\tilde{\omega}'_j+\tilde{d}_j,j=1,\ldots,n-1,\tilde{d}_{n} \}$.
Then we have
\begin{equation}\begin{split}
G_l(C,D,\widetilde{C},\widetilde{\Omega})=0, l=1,\ldots,2n-2. \label{dedu1:1}
\end{split}\end{equation}
If we further show that $G_{2n-1}(C,D,\widetilde{C},\widetilde{\Omega})$ has different signs when $d_{n}\downarrow 0$ and $d_{n}\uparrow \omega_n$, then there must exists a $d_{n}\in(0,\omega_n)$, such that $G_{2n-1}(C,D,\widetilde{C},\widetilde{\Omega})=0$.
Then our conclusion follows.

This strategy will be achieved in two steps: (i) for any given
$d_{n}\in (0,\omega_{n})$, there exist $d_i$, $i=1,\ldots,n-1$,
such that (\ref{dedu1:0}) holds for $j=1,\ldots,n-1$; and (ii)
$G_{2n-1}(C,D,\widetilde{C},\widetilde{\Omega})$ has different
signs when $d_{n}\downarrow 0$ and $d_{n}\uparrow \omega_{n}$. The
two steps can be proven similarly as in steps (i) and (iii) of
Proposition \ref{dedu2}.

Now we shall show that Case (d) holds. In this case,
$(C,\Omega)=\{(c_i,\omega_i), i=1,\ldots,n+1\}$. The proof is
similar to the proof of Case (a) in Proposition \ref{dedu2}.

Let $D=\{d_2, \ldots, d_{n}, \omega_{n+1}\}$, where $d_i\in (0,\omega_i)$, $i=2,\ldots,n$. Define
$\Omega^{\overline{D}}=\{\omega_1,\omega_2-d_2,\ldots,\omega_{n}-d_{n}\}$, $C^{-1}=\{c_2,\ldots,c_{n+1}\}$ and  $C^{-(n+1)}=\{c_1,\ldots,c_{n}\}$. We are going to show that for any given $d_{n}\in$
$(0,\omega_{n})$, there exists $d_i$, $i=2,\ldots,n-1$, $\tilde{c}_1$,  and $\tilde{c}_n$ such that
\begin{equation}\begin{split}
S^{II}_j(C^{-(n+1)},\Omega^{\overline{D}},\tilde{c}_n)=S^{I}_{j-1}(C^{-1},D,\tilde{c}_1) \label{dedu1:8}
\end{split}\end{equation}
for $j=1,\ldots,n$. Once we show that (\ref{dedu1:8}) holds, we
can define
$\widetilde{C}=\{S^{II}_j(C^{-(n+1)},\Omega^{\overline{D}},\tilde{c}_n),
j=1,\ldots, n\}$ with appropriate  $\widetilde{\Omega}$ (similar
as that of Case (a)). Then we have
\begin{equation}\begin{split}
G_l(C,\Omega,\widetilde{C},\widetilde{\Omega})=0, l=1,\ldots,2n-2. \label{dedu1:9}
\end{split}\end{equation}
If we further show that $G_{2n-1}(C,\Omega,\widetilde{C},\widetilde{\Omega})$ has different signs when $d_{n}\downarrow 0$ and $d_{n}\uparrow \omega_{n}$, then there must exist a $d_{n}\in
(0,\omega_{n})$, such that $G_{2n-1}(C,\Omega,\widetilde{C},\widetilde{\Omega})=0$.
Thus our conclusion follows.

This strategy will be achieved with the following three steps: (i)
for any given $d_{n}\in (0,\omega_{n})$ and
$\tilde{c}_{n}\in(c_{n},c_{n+1})$, there exists
$\tilde{c}_{1}\in(c_{1},c_2)$, $d_i$, $i=2,\ldots,n-1$, such that
(\ref{dedu1:8}) holds for $j=1,\ldots,n-1$; (ii) for any given
$d_{n}\in (0,\omega_{n})$, there exists $\tilde{c}_{n}\in
(c_{n},c_{n+1})$, $\tilde{c}_{1}\in(c_{1},c_2)$,  and $d_i$,
$i=2,\ldots,n-1$, such that (\ref{dedu1:8}) holds for
$j=1,\ldots,n$; (iii)
$G_{2n-1}(C,D,\widetilde{C},\widetilde{\Omega})$ has different
signs when $d_{n}\downarrow 0$ and $d_{n}\uparrow \omega_{n}$.

Define
$\tilde{c}_1=S^{II}_j(C^{-(n+1)},\Omega^{\overline{D}},\tilde{c}_n)$
for given $d_2,\ldots,d_n$ and $\tilde{c}_n$. Thus, we have
$\tilde{c}_1\in (c_1,c_2)$ and (\ref{dedu1:8}) holds $j=1$.

The proof of steps (i), (ii), and (iii) are almost exactly the
same as that of Case (a). One only needs to change the notations
and make two modifications in step (ii): first,
$\tilde{c}^*(d_{n})=S^{I}_{n-1}(C^{-1},\Omega_{d_n},c_1)$ where
$\Omega_{d_n}=\{\omega_2,\ldots,\omega_{n-1},d_n,\omega_{n+1}\}$;
second, use (viii) instead of (vii) of Proposition \ref{property}.
This completes the proof of Proposition \ref{dedu1}.
\end{proof}

\end{document}